\documentclass[a4paper,12pt]{article}
\usepackage{jheppub} 
\usepackage{lineno}
\usepackage{epigraph}
\usepackage{amsthm}
\usepackage{tikz}
\usepackage{graphicx}
\usepackage{graphics}
\usepackage{dsfont}
\usepackage{braket}
\usepackage[export]{adjustbox}
\theoremstyle{plain}
\newtheorem{definition}{Definition}[section]

\newtheorem{theorem}{Theorem}[section]
\newtheorem{lemma}[theorem]{Lemma}
\newtheorem{corollary}[theorem]{Corollary}

\newcommand{\oomit}[1]{{}}
\usetikzlibrary{shapes,arrows}
\usetikzlibrary{calc}

\usetikzlibrary {decorations.markings,fpu}
\usetikzlibrary { decorations.pathmorphing, decorations.pathreplacing, decorations.shapes, }
\usetikzlibrary{math}
\usetikzlibrary{hobby}
\usetikzlibrary{angles,quotes}
\def\lit[#1]{#1}
\def\spiral[#1](#2)(#3:#4)(#5:#6)[#7]{
  radius)[revolutions]
\pgfmathsetmacro{\domain}{#4+#7*360}
\pgfmathsetmacro{\growth}{180*(#6-#5)/(pi*(\domain-#3))}
\draw [#1,
       shift={(#2)},
       domain=#3*pi/180:\domain*pi/180,
       variable=\t,
       smooth,
       samples=int(\domain/5)] 
       plot ({\t r}: {#5+\growth*\t-\growth*#3*pi/180})
}

\newcommand{\drawdiamond}[5][1cm]{%
    \begin{tikzpicture}[baseline=(current bounding box.center)]
        \draw (0,0) -- ++(45:#1) node[right] {$#2$} -- ++(135:#1) node[above] {$#3$} -- ++(225:#1) node[left] {$#4$} -- ++(315:#1)
          node[below] {$#5$} -- cycle;
    \end{tikzpicture}
}
\newcommand{\AMSeqn}[1]{\begin{align} #1\end{align}}

\newcommand{\half}{\frac{1}{2}}

\newcommand{\abs}[1]{{\vert #1\vert}}

\newcommand{\nn}{\nonumber}

\newcommand{\lp}{LNP}
\newcommand{\lpfull}{Lamperti-Ney Process}

\title{\boldmath Conformal Dimensions On Causal Random Geometry}

\author{Ryan Barouki,}
\author{Henry Stubbs,}
\author{John Wheater}
\affiliation{Rudolf Peierls Centre for Theoretical Physics, Department of Physics, \\
Parks Road, Oxford, OX1 3PU}

\emailAdd{ryan.barouki@physics.ox.ac.uk}

\abstract{We investigate the interaction between matter and causal dynamical triangulations (CDT) in the context of
  two-dimensional quantum gravity. We focus on the Ising model coupled to CDT, contrasting this with Liouville gravity and the
  relation to the Knizhnik-Polyakov-Zamolodchikov (KPZ) formula. We demonstrate analytically for the quenched model that the
  conformal dimensions of fields on CDT align with those on a fixed lattice. We do this using a combination of lattice
methods and adapting the Duplantier-Sheffield framework to CDT, emphasizing the one-dimensional nature of CDT and its description
via a stochastic differential equation.}

\begin{document}
\maketitle
\flushbottom

\section{Introduction}
Random surfaces have been studied thoroughly in the physics literature due to their connection to string theory (the string world-sheet)
and conformal field theory. These models are also of interest from a statistical physics point of view. For example,
results of various scaling exponents of random walks on a regular $\mathbb{Z}^2$ lattice were first calculated exactly on a random
surface and then mapped to the fixed lattice \cite{duplantier1998}. In this paper we think of the random surface as defining ``two
dimensional quantum gravity''.

When defining quantum gravity, one is free to choose the class of random surfaces over which we sum in a path integral. As we will
explain in detail in section \ref{sec:discrete_models}, one natural choice comes from discretising the Einstein-Hilbert action and
``Wick rotating'' to Euclidean signature. The path integral then becomes a sum over planar maps which approximate a random
surface. We call this class \textit{Liouville random surfaces} for its connection to Liouville gravity which will be explained in
the following sections. Matter coupled to Liouville random surfaces has been studied thoroughly in the context of Liouville
gravity. One important result is the \textit{Knizhnik-Polyakov-Zamolodchikov} (KPZ) formula \cite{Knizhnik:1988ak, David:1988hj,
DISTLER1989509} which relates the scaling dimension $x$ of a conformal field on a flat background to the value $\Delta$ on a Liouville
random surface. The intuition for why scaling dimensions of fields get shifted is that the highly fractal and singular
nature of the underlying geometry (as depicted in figure \ref{fig:brownian_map}) ``dresses'' the matter degrees of freedom. For
example, the Ising model coupled to Liouville random surfaces has been calculated exactly with matrix model techniques
\cite{kazakov_ising_1986,Boulatov:1986sb} and the scaling dimensions agree with those found with
the KPZ formula. 

Liouville random surfaces bear many interesting results. One such result is the Hausdorff dimension $d_H = 4$, a dimension that
probes the large-scale structure of the space. One might consider this pathological since the aim is to define two-dimensional quantum
gravity, where two large dimensions are desirable. These issues first motivated Ambj\o rn and Loll \cite{Ambj_rn_1998} to consider
another class of random surfaces -- those which are globally hyperbolic i.e. contain a time slicing. Time slicing is equivalent
to including the causal structure in the set of allowed surfaces. This program is thus called \textit{causal dynamical
triangulations} (CDT) since it is often implemented as a sum over discrete triangulations. 

In this paper, we are interested in how matter couples to CDT -- we use the Ising model as an example but the main results can be
extended to so-called RSOS, or height models.
Not much is known analytically, but numerical studies suggest that
the scaling dimensions of the Ising fields are not shifted from their classical values.
It was shown using Monte Carlo methods and a high-temperature expansion that the critical exponents of a single Ising model
($c=1/2$) coupled to CDT appear to retain their classical values \cite{Ambj_rn_1999_numerics}. Furthermore, a numerical study of the three-state Potts model ($c=4/5$) found the same behaviour \cite{Ambj_rn_2009}. Perhaps more surprisingly, in a numerical study of 8 Ising
models ($c=4$) coupled to CDT, the scaling dimensions of the matter fields still remained unchanged despite a change in the
gravity sector \cite{Ambj_rn_2000}. It appears that CDT coupled to unitary matter has the universal property that the critical
exponents take their classical values and that this result is robust even beyond $c=1$. Interestingly, it has been found that the scaling exponents \emph{do} shift when coupling hard dimers to CDT in a certain phase of the model \cite{Wheater_2022}. However, the hard dimer model at its critical point has a central charge $c=-22/5$ and so is non-unitary. The model has complex weights
and so cannot be described in terms of a positive definite stochastic process. The results of this paper do not apply to such models.

We provide a series of analytical arguments for why the scaling dimensions of fields on CDT are no different to those on a fixed
lattice for the quenched model. Crucial to our arguments are the topological defect formulation of the Ising model \cite{Aasen_2016} and evasion of the KPZ relation. The most useful formulation of the latter for our purposes is that of
Duplantier-Sheffield (DS) \cite{duplantier2010liouville} whose framework we adapt to CDT 
to show that no analogue of KPZ applies. This follows from the fact that CDT is essentially one-dimensional and can be described in terms of a
stochastic differential equation (SDE) which has continuous sample paths \textit{almost surely}. As  a consequence, KPZ is evaded.

This paper is structured as follows. In section \ref{sec:discrete_models} we outline the essential theoretical background for
models of 2D quantum gravity, including CDT and the DS approach to Liouville gravity. In section \ref{sec:top_ising} we explain
the plaquette formalism of the Ising model as described in \cite{Aasen_2016} and use the algebra of topological defects to find algebraic relations for operators that will be crucial for finding the conformal dimensions: the Dehn twist operators. In section \ref{sec:stochastic} we show that one can indeed construct a Dehn twist in the continuum, from which we calculate the conformal dimensions of the Ising fields. Furthermore, we provide an alternative proof that there will be no KPZ-like relation in continuum CDT by constructing a random measure and following the arguments of DS in \cite{duplantier2010liouville}. In section
\ref{sec:horava} we extend the existing connection to Ho\v rava--Lifshitz gravity originally found in \cite{Ambj_rn_2013}.
Finally, in section \ref{sec:annealed} we provide an argument for why our results may apply more generally to the model where
the geometry and matter are sampled according to a joint measure (the \textit{annealed model}), which relies solely on the
continuity of the process describing the evolution of CDT.
\section{2D quantum gravity}
\label{sec:discrete_models}
In Lorentzian space-time, quantum gravity on the two dimensional manifold $M$ is defined by the path integral
\begin{equation}
    Z(\Lambda) = \int \mathcal{D}[g]\, \mathcal{D}[\Phi]\,e^{iS_{EH}(g)+iS_m(\Phi,g)}\,.\label{eqn:LorentzianPI}
\end{equation}
Here $[g]$ denotes the equivalence class of Lorentzian metrics $g$ on $M$ up to diffeomorphisms, and $[\Phi]$ the configuration of, for the moment unspecified, matter degrees of freedom living on $M$.
$S_{EH}(g)$ is the Einstein-Hilbert action
\begin{equation}
    S_{EH}(g) = \frac{1}{16\pi G_N}\int_M d^2x \sqrt{g}(2\Lambda - R)\,,
\end{equation}
where $G_N$ is Newton's constant, $\Lambda$ is the cosmological constant, and $R$ is the Ricci scalar curvature. $S_m(\Phi,g)$ is the diffeomorphism invariant action for the matter degrees of freedom. 
For simplicity we have omitted the Gibbons-Hawking-York term for manifolds with boundary. We will work with the Euclidean gravity model which is similarly defined by
\begin{equation}
    Z(\Lambda) = \int \mathcal{D}[g]\, \mathcal{D}[\Phi]\,e^{-S_{EH}(g)-S_m(\Phi,g)}\,,\label{eqn:EuclideanPI}
\end{equation}
where now the functional integration is over Euclidean, rather than Lorentzian, metrics. Note, however, that the relationship between Lorentzian and Euclidean models cannot be simply seen as 
the  Wick rotation $t\to -it$, as for field theory in flat space-time, because this transformation does not commute with diffeomorphisms.

In two dimensions, Euclidean manifolds are completely characterised by their genus $g$ and the number of
boundaries $b$. Moreover, the curvature part of the Einstein-Hilbert action is a topological invariant due to the
Gauss-Bonnet theorem
\begin{equation}
  \label{eq:euler_char}
    \int_M d^2x\sqrt{g}R = 4\pi\chi(M),
\end{equation}
where $\chi(M) = 2 - 2g - b$ is the Euler characteristic of the manifold. Hence for a fixed topology, every geometry receives the
same factor (\ref{eq:euler_char}) in the path integral and can be factored out leaving only the cosmological constant term, so we henceforth set $8\pi G_N=1$.  The physics of the model then lies largely in the proper definition of the functional integral over the metric. We will be concerned with two (apparently) different methods for doing this: the first is by the explicit discretization of spacetime; and the second by representation in terms of stochastic processes.

\subsection{Graphs}

Much of this paper is most conveniently expressed in the language of graphs so for convenience we gather our definitions here.
\begin{definition}
\item (i) A simple planar graph $G = (V,E)$ consists of a set of vertices $V$ and a set of edges $E$ that are unordered pairs of distinct vertices $\langle u, v\rangle$, such that no two vertices share more than one edge, and such that $G$ can be embedded in the plane with no two edges intersecting. We denote by $d(u,v)$ the graph distance between vertices $u$ and $v$.

\item  (ii)  A rooted graph $G$ is a graph with one marked vertex called  the \textit{root} $v_0$.

\item (iii) The dual graph $G^*$ is constructed from $G$ by placing a vertex in every face, and joining the new vertices with edges such that each new edge intersects only one edge of $G$.
Essentially, vertices are mapped to faces, edges to edges and faces to vertices.
\item (iv)  A \textit{triangulation} $T$ is a rooted planar graph where every face is a triangle.  
\item (v)  A \textit{tree} is planar graph with no cycles.
  \end{definition}

\subsection{Euclidean triangulations}
One way to implement the gravitational path integral \eqref{eqn:EuclideanPI} is to discretise spacetime and replace the functional integral by a sum over discrete geometries. This
approach provides a UV cut-off in the form of the lattice
spacing and is naturally non-perturbative. Indeed, the use of generalized triangulations to approximate the
geometry of manifolds in general relativity goes back to the work of Regge \cite{Regge1961GeneralRW} on coordinate-independent 
computational methods. 
Here we focus on the approximation of a two-dimensional manifold $M$ and its geometry $g$ by
triangulations composed of  piece-wise flat equilateral triangles each of the same area $\alpha$. Heuristically, a given continuum configuration of fixed area $A$ can approximated progressively more accurately by increasing the number of triangles, $N_\alpha$, used and decreasing $\alpha$  such that 
$N_\alpha \alpha=A$ is held fixed. Additionally, for fixed $N_\alpha$, a given triangulation defines a set of graph distances $\{d (v_i, v_j)\}$ between vertices pairwise that approximates the set of geodesic distances on $M$; hence each distinct triangulation approximates a distinct metric $g$.

We first consider the pure gravitational system with no matter degrees of freedom on a disk $M$.  
Then 
let  $\mathcal E (M)$ be the set of inequivalent, unrestricted triangulations of $M$ and define the partition function
\begin{equation}
W[\tilde g,z]=\sum_{T\in \mathcal E (M)} \tilde g^{\Delta(T)}z^{\abs{\partial T}} \,.\label{eqn:Wdefinition}
\end{equation}
Here 
we denote by $\Delta(T)$ the number of triangles in $T$, and by $\abs{\partial T}$ the number of edges in the boundary of $T$. The combinatorics of  $\mathcal E (M)$  was first elucidated by Tutte in \cite{Tutte1962a, Tutte1962b, Tutte1962c, Tutte1963}. Crucially the number of triangulations $T\in \mathcal E (M)$ such that $\Delta(T)=n$ grows more slowly than $\tilde g_c^{-n}$ with $0<\tilde g_c<1$. It follows that the sum in \eqref{eqn:Wdefinition} converges and $W$ exists for $\tilde g<\tilde g_c$.

It was proposed in \cite{KAZAKOV1985295,DAVID1985543} that the sum over triangulations in \eqref{eqn:Wdefinition}, taken in the limit $\tilde g\uparrow\tilde g_c$ where arbitrarily large triangulations dominate,  can be identified with the Euclidean functional integral over the metric. Setting  $\tilde g \to
\tilde g_c e^{-\Lambda \alpha}$, identifying the area 
\begin{equation}
\label{eqn:vol_discrete}
   \alpha \Delta(T)\to \int_\mathcal{M} d^2x \sqrt{g} \,,
\end{equation}
and the triangulation sum
\begin{equation}
\sum_{T\in \mathcal E (M)} \frac{1}{C(T)}\tilde g_c^{\Delta(T)} (.)\to  \int \mathcal{D}[g] (.)\,,\label{eqn:measureT}
\end{equation}
we see that $W[\tilde g,1]$ formally reproduces $Z[\Lambda]$ \eqref{eqn:EuclideanPI} (the curvature term can be ignored as it is topological).
There is now a great deal of evidence that this identification is correct and that the Euclidean functional integral over the metric can indeed be defined as a sum over triangulations. Much of this evidence comes through the connection between $\mathcal E (M)$ and matrix models which was established in  \cite{Brezin:1977sv}; for a modern review see \cite{Anninos:2020ccj}.

Each $T\in \mathcal E (M)$ is a random planar map. This ensemble  has been studied in great detail; one particularly interesting result is that
the Hausdorff dimension of very large planar maps is $d_H = 4$ almost surely \cite{gall_hausdorff}. Thus such a map  must have very non-trivial fractal structure since it is assembled piecewise from two-dimensional triangles (or, more generally, polygons). Figure \ref{fig:brownian_map} shows an example of a
planar map with a large number of vertices; it is a discrete approximation of the Brownian map \cite{gall_brownian_map}. We will return to the characterization of such maps in section \ref{sec:StochasticLiouville}.
\begin{figure}
  \begin{center}
    \includegraphics[width=0.95\textwidth]{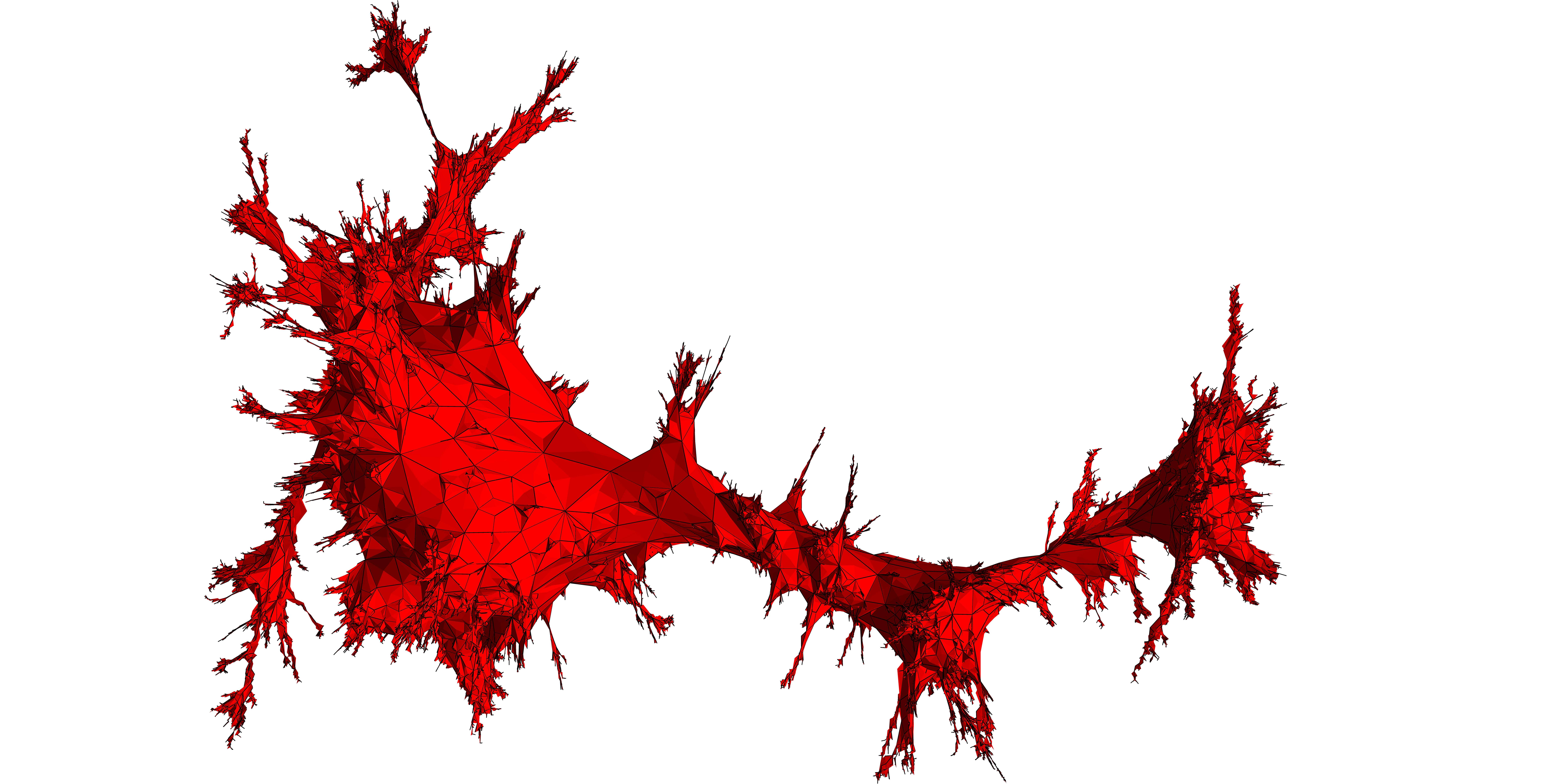}
  \end{center}
  \caption{An instance of a large planar map with 30,000 vertices generated by J\'er\'emie Bettinelli \cite{JerBet}.}
  \label{fig:brownian_map}
\end{figure}

\subsection{Causal triangulations}

The Euclidean triangulations provide a definition of the Euclidean path integral but it is not immediately clear how results in this theory are related to physical quantities in the Lorentzian theory \eqref{eqn:LorentzianPI}. This led Ambj{\o}rn and Loll \cite{Ambj_rn_1998} to introduce the causal  triangulation (CT). The CT  triangulates only manifolds that are globally hyperbolic, i.e. they have a time slicing. Consequently the edges in a CT are either time-like or space-like which allows a well-defined Wick rotation between Euclidean and Lorentzian signatures and the notion of causality is automatically incorporated. The model is most often solved in the Euclidean setting but a consistent way to define Lorentzian triangles and  angles that satisfy the usual Euclidean additivity
conditions exists and was described in \cite{Sorkin_1975}.

\begin{definition}\label{def:CT}
    A \textit{causal triangulation} (CT) $C$ of the disk $M$ is a triangulation where all vertices at a fixed graph distance $t$ from the root form a cycle $S_t(C)$.
    Its height $h(C)$ is the maximum graph distance of the vertices from the root
    $$
       h(C) = \max_{v\in C}\{d(v,v_0)\}.
    $$
\end{definition}
\noindent An example is shown in figure \ref{fig:tree_bijection}.
It is clear from the definition that each vertex of $C$ at height $t>0$ has 2 spatial neighbours, several future neighbours $\sigma_v^f$ and a number
of past neighbours $\sigma_v^p$ such that the total number of neighbours is $\sigma_v = 2 + \sigma_v^f + \sigma_v^p$. The definition is easily extended to annular topology by marking a vertex at height $t=1$ then deleting the root and the edges attached to it. 

Letting  $\mathcal C (M)$ be the set of distinct causal triangulations of $M$ we define the partition functions
\begin{equation}
W[\tilde g,z;t]=\sum_{C\in \mathcal C (M): \,h(C)=t} \abs{S_t(C)}\, \tilde g^{\Delta(C)+1}(z/\tilde g)^{\abs{S_t(C)}} \,.\label{eqn:WCTdefinition}
\end{equation}
Note that, in contrast to the Euclidean triangulation case, we now have a natural `time' $t$ which will play an important role in what follows.  From the gravitational point of view, $W$ is essentially the discretized path integral for the Euclidean amplitude that a universe evolves from a point to a boundary in time $t$. There is a critical value $\tilde g_c=\half$ such that the sum in \eqref{eqn:WCTdefinition} converges for $\tilde g<\tilde g_c$ for all $t$; in the limit $\tilde g\uparrow\tilde g_c$ arbitrarily large triangulations dominate and we recover the path integral for the point-to-boundary amplitude in two-dimensional projectable Horava-Lifshitz gravity \cite{Ambjorn:2013joa}.

Although the sum to determine $W[\tilde g,z;t]$ can be evaluated by elementary means for finite $t$ \cite{Ambj_rn_1998}, this gives us little information about the nature of the CTs that contribute at criticality. It was observed in \cite{tree_bijection} that 
 $\mathcal C(M)$ can also be studied by exploiting the existence of a  bijection $\beta : \mathcal{C} \to \mathcal{T}$  with the set of planar rooted trees, see \cite{durhuus2022trees} for a recent review.
\begin{definition}\label{def:bijection} The bijective map $\beta(C)$ is defined by:
\begin{enumerate}
    \item Mark an edge coming out of $v_0$ and add a root vertex $r$ to $C$ connected by a single edge to $v_0$ placed in the face
      to the right of the marked edge as seen from $v_0$.
    \item Remove all edges in the cycles $S_t(C), \forall t > 0$. i.e. the space-like edges.
    \item For every $v \in S_t(C)$ where $1 \leq t < h(C)$ remove the rightmost edge of the $\sigma_v^f$ edges as seen from
      $v$.
\end{enumerate}
\end{definition}
The resulting graph $T$ is a tree with a new root vertex $r$ and every vertex from $C$ as shown in figure \ref{fig:tree_bijection}.
\begin{figure}
  \begin{center}
    \includegraphics[width=0.9\textwidth]{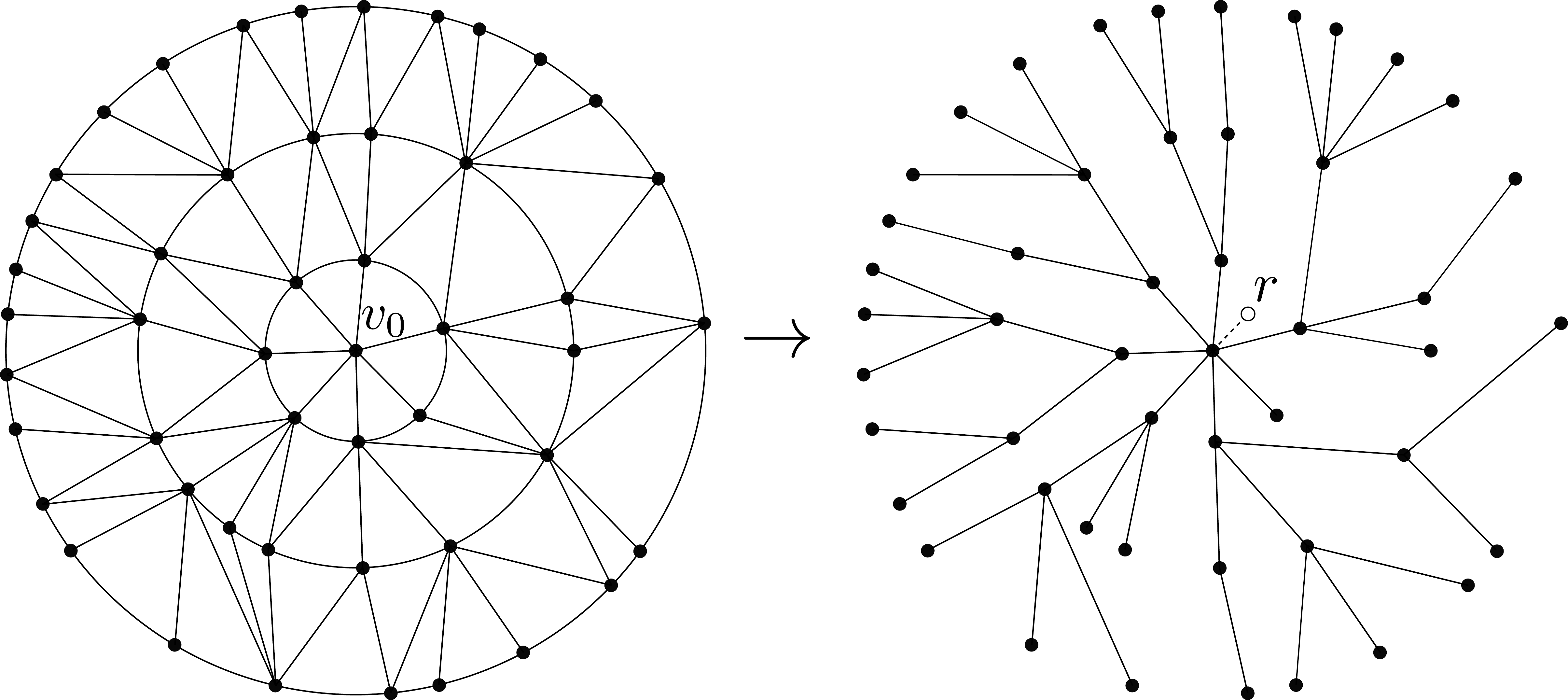}
  \end{center}
  \caption{A causal triangulation of height 3, and  the action of the map $\beta$ to its associated tree.}\label{fig:tree_bijection}
\end{figure}
It is straightforward to show that the inverse map $\beta^{-1}$ exists. The partition functions $W$ can be rewritten in the tree picture as
\begin{equation}
W[\tilde g,z;t]=\sum_{T\in \mathcal T: \,h(T)=t+1}\,\prod_{v\in V(T)\backslash r} \tilde g^{\sigma_v}
(z/\tilde g)^{\abs{S_t(\beta^{-1}(T)}} \,.\label{eqn:WCTtreedefinition}
\end{equation}

It was shown in \cite{Durhuus:2009sm} 
that the ensemble \eqref{eqn:WCTtreedefinition} at criticality, i.e. $\tilde g =\half$, is equivalent to a critical Galton-Watson (GW) process. 
This is a branching process specified by a set of probabilities
$p_n, n=0,1,...$, called the \textit{offspring probabilities}, having unit mean i.e., $\sum_{n=0}^\infty np_n = 1$. 
The process can be viewed as a tree in which $p_n$ is the probability of a single vertex having $n$ offspring.
Recalling that by definition the root vertex has a single offspring with probability one, 
the probability distribution on finite trees $T\in\mathcal T$ is then given by
\begin{equation}
\label{eq:prob_measure}
    \pi(T) = \prod_{i\in T \setminus r} p_{\sigma_i-1}.
\end{equation}
Choosing $p_n=2^{-n-1}$ we see that $\pi(T)$ reproduces the weights in $W[\half,1]$; they define the generic random tree ensemble \cite{Durhuus:2009sm}. It is convenient to introduce the generating function for 
offspring probabilities
\begin{equation}
    f(s) = \sum_{n=0}^\infty p_n s^n.\label{eqn:fsdefinition}
\end{equation}
Then  $f(1) = 1$, $f'(1) = 1$, and for the generic random tree $f(s)=(2-s)^{-1}$.  

The definition of CTs \ref{def:CT} is easily extended to triangulations that have  infinite height  $h(C)=\infty$ and no boundary; these form the set $\mathcal C_\infty$ and cover the plane. 
Applying the bijection $\beta$ \ref{def:bijection}  to $C\in\mathcal C_\infty$ then yields an infinite tree $T\in\mathcal T_\infty$. 
It was shown in \cite{durhuus2022trees} that one can extend the probability distribution $\pi(T)$ on finite trees to a measure $\nu(T)$
on $T\in\mathcal T_\infty$. 
Provided only that $f''(1) < \infty$, $\nu(T)$   is concentrated on single spine trees \cite{durhuus2022trees,
durhuus_spine2003}. 
The spine is a sequence of vertices $r, s_1, s_2, ...$ starting at the root and going off to infinity with no
backtracking. Attached to each spine vertex, with probability $p_n'$ generated by $sf'(s)$, is a set of $n$ finite trees 
with vertices distributed according to $f(s)$.  The ensemble defined by the set of triangulations  $\mathcal C_\infty$ and the probability measure on them $\mu(C)=\nu(\beta^{-1}(C))$   is called the Uniform Infinite Causal Triangulation (UICT). 

\subsection{Coupling matter}
\label{sec:CouplingMatter}

To implement the matter action $S_m$ \eqref{eqn:EuclideanPI} in the discretized picture we introduce new degrees of freedom that live on the vertices of the graphs $G$ in the ensemble. In this paper we work specifically with Ising spins but the construction is easily generalised, in particular to the random height models. 
On each vertex $v\in V(G)$ lives a spin that takes values in $\{+1,-1\}$, and spins on vertices that share an edge $e\in E(G)$ interact with a coupling strength $J_e$.
The spin partition function on $G$ is  given by  the sum over all spin configurations 
\begin{equation}\label{eqn:IsingZ}
    Z_G[\{J\}_G] = \sum_{\{\sigma\in \{+1, -1\}^{V(G)}\}}\exp\left(\sum_{e = \langle x, y \rangle \in E(G)}J_e\sigma_x\sigma_y\right).
\end{equation}
From now on we will assume that for the $\mathcal E$ ensemble $J_e=J, \forall e$, and that for the $\mathcal C$ ensemble $J_e=J_1$ for all space-like edges and $J_e=J_2$, for all time-like edges.
The discretised Euclidean and CT path integrals for the disk corresponding to \eqref{eqn:EuclideanPI} then take the form
\begin{equation}
W_{\mathcal G}=\sum_{G\in\mathcal G} \tilde g^{\Delta(G)} z^{\abs{\partial G}}  Z_G[J_{\mathcal G}]\,,
\end{equation}
where $\mathcal G$ is chosen to be respectively $\mathcal E$ or $\mathcal C$. We refer to the systems described by $W_{\mathcal G}$ as \emph{annealed} models.

In statistical mechanics language $W_{\mathcal G} $ is the grand canonical partition function. For given $J_{\mathcal G}$ the sum is convergent for $\tilde g <\tilde g_c(J_{\mathcal G})$.
In the coupling strength   region $J_{\mathcal G}\in\mathcal A$ where 
 $\tilde g_c(J_{\mathcal G})$ is  analytic, 
 the limit $\tilde g \uparrow \tilde g_c(J_{\mathcal G})$ describes the same purely gravitational physics as the partition function without matter degrees of freedom. On the other hand, in the region $J_{\mathcal G}=J^*_{\mathcal G}\in\partial \mathcal A$ where  $\tilde g_c(J_{\mathcal G})$ is not analytic 
the spins become critical on the very large graphs in the sum  -- they magnetize in the case of the Ising model. The limit $\tilde g \uparrow \tilde g_c(J^*)$ then describes a continuum theory of gravity interacting with matter (central charge $c=\half$ conformal matter in the Ising case). The matrix model solution \cite{kazakov_ising_1986,Boulatov:1986sb} for $W_{\mathcal E}$ shows that the Ising scaling exponents are shifted away from their regular $\mathbb Z^2$ lattice values; they are related to each other by the KPZ formula which can be understood in a number of ways as is discussed  below in section \ref{sec:StochasticLiouville}.
Numerical simulations  and series expansions  for ${\mathcal C}$ \cite{Ambj_rn_1999_numerics}
strongly  suggest that the scaling exponents for Ising spins are not shifted from their regular lattice values, but no exact solution for $W_{\mathcal C}$ is known.

An alternative formulation of the CT interacting with matter is provided by working in the canonical ensemble with the graphs $C\in\mathcal C_\infty$. In principle we have a new measure  $\mu_J(C)$ that reflects the relative weight of the Ising partition function $Z_C(J)$ on different graphs. As critical behaviour of the spins only occurs on very large graphs we might expect $\mu_J(C)$ to capture the same physics as the grand canonical partition function $W_{\mathcal C}$; unfortunately  $\mu_J(C)$  is yet to be constructed. An intermediate step is to analyse the critical properties of an Ising spin system living on a triangulation sampled from the UICT ensemble according to the measure $\mu$, which we  call the \emph{quenched} model. In \cite{tree_bijection} it was shown that at small $J$ this system has a unique Gibbs measure (corresponding to unmagnetized spins), while at large $J$ (at least) two Gibbs measures co-exist (corresponding to two possible magnetised states); furthermore it was shown that almost surely (i.e. with probability one in the measure $\mu$) the critical temperature is the same for any $C$. The main purpose of this paper is to show that \emph{if} this critical point leads to a scaling limit, \emph{then} the corresponding field theory 
must contain operators with the same scaling exponents as those appearing in the scaling limit of the flat lattice Ising model.

\subsection{Liouville gravity \`a la Duplantier-Sheffield}\label{sec:StochasticLiouville}
The Liouville gravity approach to computing the gravitational path integral was introduced by Polyakov \cite{POLYAKOV1981207}. 
We outline the basic features here (see
\cite{Teschner_2001, Nakayama:2004vk} for a detailed review of the topic). 
Any 2D metric can be written as $g=e^{\phi}\hat{g}$, where $\phi$ is a dynamical degree of freedom and
$\hat{g}$ is a fixed background metric which can always be taken to be flat Euclidean space, assuming a suitable topology. The Euclidean Einstein-Hilbert action reduces to the Liouville action for the scalar field
$\phi$ given by
\begin{equation}
    S_L = \int d^2z \sqrt{\hat{g}}(\hat{g}^{ab} \partial_a \phi \partial_b \phi + Q\hat{R}\phi + \Lambda e^{\gamma\phi}) +
    S_{CFT}\,.
\end{equation}
Here $\hat{R}$ is the Ricci curvature associated with $\hat{g}$,  $\Lambda$ is the cosmological constant and $S_{CFT}$ is the
action of some conformal matter with central charge $c$. 
The requirement that there is no conformal anomaly then determines 
$Q = \frac2\gamma + \frac\gamma2$, where $\gamma$ is given by
\begin{equation}
  \label{eq:gamma_c}
    \gamma = \frac{1}{\sqrt{6}}(\sqrt{25-c} - \sqrt{1-c}).
\end{equation}
For example, pure gravity ($c=0$)
corresponds to $\gamma = \sqrt{8/3}$ whereas the scaling limit of the Ising model ($c=\half$) corresponds to $\gamma = \sqrt{3}$. 

The Liouville path integral for pure gravity, taking background metric $\hat g_{ab} = \delta_{ab}$, is written formally as
\begin{equation}
  Z[\Lambda] = \int \mathcal{D}\phi \,e^{-\int d^2z\, \partial_a\phi \partial^a\phi + \Lambda e^{\gamma\phi}}\,.
  \label{eq:liouville_path_int}
\end{equation}
To make rigorous sense of (\ref{eq:liouville_path_int}) we must define 
the measure on
$\phi$. In the probabilistic approach to QFT, the idea is that the term $e^{-\int d^2z \,\partial_a\phi \partial^a\phi}$ is
proportional to a Gaussian measure on the space of functions $\phi: D \to \mathbb{R}$ with Dirichlet boundary conditions where $D
\subset \mathbb{R}^2$ \cite{Simon_1975}. This ensemble is called the \textit{Gaussian free field} (GFF); it is
a generalisation of Brownian motion to higher dimensions and is defined as follows.

\begin{definition} 
Let $D \subset \mathbb{R}^d$ be some domain and define the inner product on functions $f: D \to
\mathbb{R}$ as
\begin{equation}
    \langle f,g \rangle_\nabla := \frac{1}{2\pi}\int_D \nabla f \cdot \nabla g\,d^dz\, ,
\end{equation}
and associated norm
\begin{equation}
    ||h||_\nabla^2 := \frac{1}{2\pi}\int_D \nabla h \cdot \nabla h\,d^dz.
\end{equation}
The GFF  is defined to be the measure  whose probability density is given by
\begin{equation}
    \rho(h) = const.\, \exp{\left(-\frac12 ||h||_\nabla^2\right)}.
\end{equation}
\end{definition}
\noindent The density $\rho(h)$  is nothing but the path integral measure for the massless free boson with Dirichlet boundary conditions. Hence 
 the formal path integral (\ref{eq:liouville_path_int}) is defined rigorously as
\begin{equation}
  \int \mathcal{D}\phi\, e^{-\int d^2z \partial_a\phi \partial^a\phi + \Lambda e^{\gamma\phi}} = \mathbb{E}[e^{-\int d^2z \Lambda
  e^{\gamma\phi}}],
\end{equation}
where the expectation is over the GFF measure.

In this formulation of Liouville gravity, the cosmological constant couples to the random area
\begin{equation}
    A = \int e^{\gamma\phi} d^2z\,,
\end{equation}
generated by the GFF $\phi$.
On the common base space $[0,1]^2$ there are then two measures:
\begin{itemize}
    \item The Lebesgue measure on $[0,1]^2$, $d^2z$;
    \item The random measure $d\mu_\gamma= e^{\gamma \phi} d^2z$, where $\phi$ is a GFF.
\end{itemize}
To visualize the random measure $\mu_\gamma$, choose $\delta\in(0,1) $ and, starting with  the base space $[0,1]^2$, iteratively divide squares into four quadrants to obtain the set of largest square regions $S_i \subset [0,1]^2$ such that $\mu_\gamma(S_i) \leq \delta$.
This defines
the \textit{dyadic square decomposition} where each square has roughly the same quantum area $\delta$. 
The decomposition is shown in figure
\ref{fig:gff_dyadic} 
for a particular instance of the GFF. 
\begin{figure}[h]
  \begin{center}
    \includegraphics[width=0.45\textwidth]{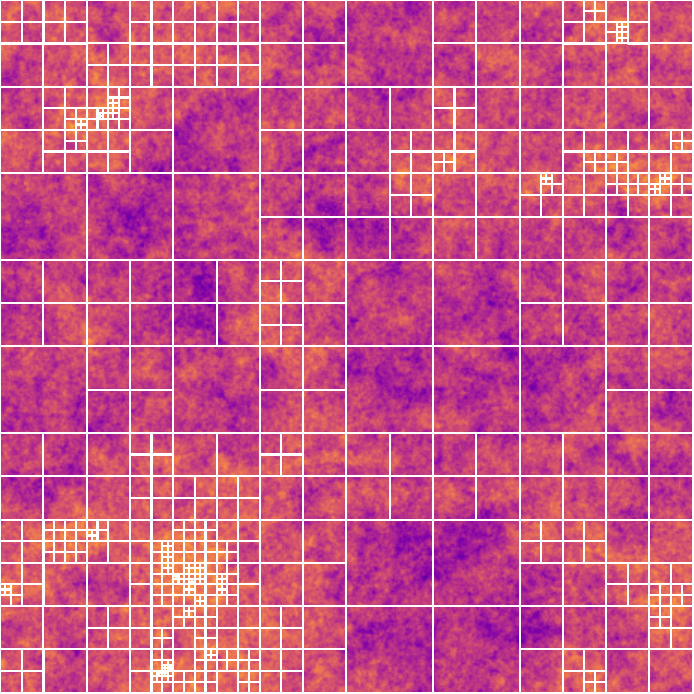}
    ~~~
    \includegraphics[width=0.45\textwidth]{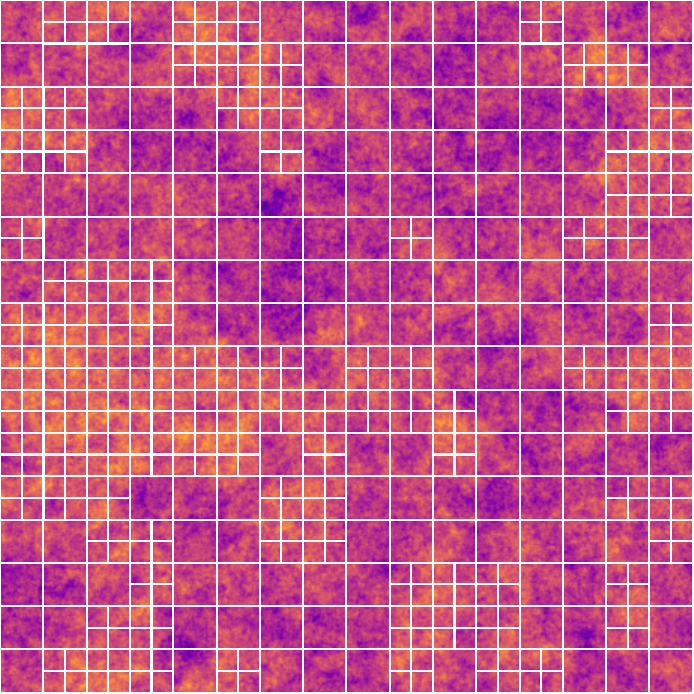}
  \end{center}
  \caption{A dyadic decomposition of $e^{\gamma\phi}$ with $\gamma = 2$ (left) and $\gamma=0.5$ (right)}
  \label{fig:gff_dyadic}
\end{figure}
The strength of the fluctuations is determined by $\gamma$, which is in turn determined by the type of matter on the background through (\ref{eq:gamma_c}).

Duplantier and Sheffield
\cite{duplantier2010liouville} observed that  a subset $K \subset [0,1]^2$ can be measured by using either the Lebesgue measure or the random measure. 
The random measure with fixed $\phi=0$, \emph{i.e.} no fluctuations in the 2D metric, is simply the Lebesgue measure, so the relationship between the two results characterises the effect of gravity on $K$.
To make this relationship precise first  define two kinds of balls:
\begin{definition}[Euclidean and quantum balls] For all $z \in [0,1]^2$ {\rm{:}}
\begin{itemize}
    \item $B_\epsilon(z)$ is the Euclidean ball of radius $\epsilon$ centred on $z$;
    \item $B^\delta(z)$ is the quantum ball centred on $z$ and of quantum area $\delta$ i.e. It is the Euclidean ball $B_\tau(z)$ with $\tau := \sup
      \{r \geq 0,~ \mu_\gamma(B_r(z)) \leq \delta \}$.
\end{itemize}    
\end{definition}
\noindent We can then define two different  scaling exponents for some fixed subset $K \subset [0,1]^2$:
\begin{definition}[Euclidean and quantum scaling] ~
\begin{itemize}
    \item The \textbf{Euclidean scaling exponent} $x = x(K)$ is defined as
        \begin{equation}
          \label{eq:classical_exponent}
            x(K) := \lim_{\epsilon \to 0} \frac{\log \mathbb{P}[B_\epsilon(z) \cap K \neq \emptyset]}{\log \epsilon^2}\,,
        \end{equation}
        where $z$ is sampled according to, and the probability $\mathbb{P}$ computed in, the Lebesgue measure;
    \item The \textbf{quantum scaling exponent} $\Delta = \Delta(K)$ is defined as
        \begin{equation}
            \label{eq:quantum_exponent}
            \Delta(K) := \lim_{\delta \to 0} \frac{\log \mathbb{E}[\mu_\gamma[B^\delta(z) \cap K \neq \emptyset]]}{\log \delta}\,,
        \end{equation}
        where $z$ is sampled according to, 
        and the expectation $\mathbb{E}$ computed in, 
        the random measure $\mu_\gamma$.
\end{itemize}
\end{definition}
The main result of Duplantier-Sheffield \cite{duplantier2010liouville} is the derivation of the
\textit{Knizhnik-Polyakov-Zamolodchikov} (KPZ) formula \cite{Knizhnik:1988ak} that relates these two scaling dimensions by
\begin{equation}
\label{eq:KPZ}
    x = \frac{\gamma^2}{4}\Delta^2 + (1-\frac{\gamma^2}{4})\Delta.
\end{equation}
Taking the Ising model as an example, we have $\gamma = \sqrt{3}$ and 2 primary fields $\varepsilon, \sigma$ with
scaling dimensions $x_\varepsilon = 1/2, x_\sigma = 1/16$.
Using (\ref{eq:KPZ}) we find that the quantum scaling dimensions take the shifted values
\begin{equation*}
    \Delta_\varepsilon = 2/3, ~~ \Delta_\sigma = 1/6.
\end{equation*}
These are in agreement with the matrix model calculations of Kazakov et al \cite{kazakov_ising_1986, Boulatov:1986sb}.
\section{The Ising model and topological defects}
\label{sec:top_ising}

The exact microscopic connection between topological defects in the two-dimensional Ising model  and the corresponding fusion category was elucidated in \cite{Aasen_2016}, and then extended to general height models in \cite{aasen2020topological}. In this section we review the formalism for the Ising model given in \cite{Aasen_2016}, and describe in detail how it can be applied to CT graphs of disk or annulus topology. 

\subsection{The Ising model in the plaquette formalism}
Consider the  Ising model with partition function $Z_G[\{J\}_G]$   \eqref{eqn:IsingZ} defined on a planar graph $G$. 
Now construct a new graph $\widetilde G$ formed by combining $G$ with its dual $G^*$ as follows: start with $G$ 
and $G^*$ overlaying each other, then join each $G$ vertex to the nearest $G^*$ vertices without crossing any original or dual edges, and finally remove the original and dual edges. 
See figure \ref{fig:lattices} for the case when $G$ is a square lattice. 
\begin{figure}
  \begin{center}
    \includegraphics[width=0.5\textwidth]{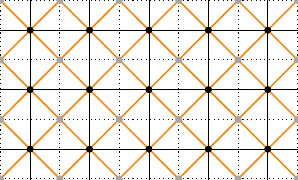}
  \end{center}
  \caption{The graph $\widetilde G$ (orange) formed from the original square lattice (black), and its dual (gray, dashed).
  }\label{fig:lattices}
\end{figure}
The graph $\widetilde G$ is formed of quadrilateral faces with opposite vertices belonging to the original or dual map. We call
these quadrilateral faces \textit{plaquettes}. It is important to note that the spins  exist either on $G$  or on $G^*$, but not both. Therefore, each plaquette only has 2 spins. 
Each plaquette represents an edge between spins on the original lattice. 

If $G$ is a CT, we have two types of plaquettes in $\widetilde G$, \textit{horizontal} and \textit{vertical} as shown in
figure \ref{fig:cdt_plaquettes}. By assigning weights to each plaquette we can rewrite the partition function as a product of
these weights. First, we define $u_H$ and $u_V$ via

\begin{figure}[h]
  \begin{center}
    \includegraphics[width=0.5\textwidth]{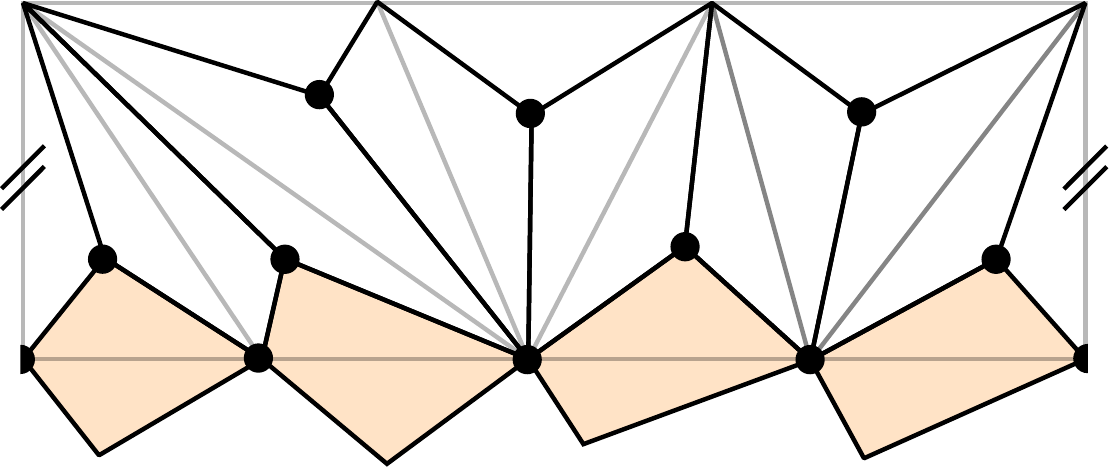}
  \end{center}
  \caption{Horizontal (orange) and vertical plaquettes for a single CDT strip where the dots show the location of the weights
  $d_v$.}\label{fig:cdt_plaquettes}
\end{figure}

\begin{equation}
    e^{2J_V} = \cot u_V,~~e^{2J_H} = \cot(\frac{\pi}{4} - u_H).
\end{equation}
where $J_H$ and $J_V$ are the couplings for horizontal and vertical edges respectively.
 It is useful in the context of defects to re-define the partition function to include a vertex factor $d_v$ 
\cite{Aasen_2016}
\begin{equation}
  \label{eqn:quantum_dims}
    d_v = \begin{cases}
        1 & \text{if $v$ has a spin} \\
        \sqrt{2} & \text{if $v$ is empty.}
    \end{cases}
\end{equation}
We assign these weights to the plaquettes by splitting them between the left and right vertices. Since
the vertical plaquettes only have spins on the top and bottom positions this gives an additional $\sqrt{2}$ factor. Horizontal
plaquettes have spins on the left and right positions so the additional factor from the weights is 1. 
We then define the weights for each plaquette 
\begin{equation}
\label{eq:plaquettes}
\begin{split}
    W_j^V(u_V)\equiv \drawdiamond[0.7cm]{}{a}{}{b} &= \sqrt{2}\,(\cos u_V \delta_{ab} + \sin u_V \sigma^x_{ab})  = \begin{cases}
        \cos u_V & a=b \\
        \sin u_V & a \neq b
    \end{cases} \\
    W_j^H(u_H)\equiv \drawdiamond[0.7cm]{b}{}{a}{} &= \frac{1}{\sqrt{2}}\,(\cos u_H + (-1)^{a+b}\sin u_H) = \begin{cases}
        \cos (\frac\pi 4 - u_H) & a=b \\
        \sin (\frac\pi 4 - u_H) & a \neq b
    \end{cases}
\end{split}
\end{equation}
where the labels are related to the spins by $\sigma = (-1)^a$. The partition function is now a product of plaquettes
\begin{equation}
    Z_G[J_V,J_H] = \sum_{\{\sigma\}} \prod_{p \in {\mathcal{P}_{\widetilde G}}} \drawdiamond[0.7cm]{c_p}{b_p}{a_p}{d_p},
\end{equation}
where $\mathcal{P}_{\widetilde G}$ denotes the set of plaquettes in $\widetilde G$.

\subsection{The spin-flip defect}
One can define an Ising model in the presence of a spin defect. A spin defect is a one-dimensional object that bisects a sequence
of edges in the original lattice and swaps the couplings $J\to-J$, turning the interaction from ferromagnetic to
anti-ferromagnetic and vice versa. We can implement this defect in the plaquette formalism by cutting a path along $\widetilde G$ and splitting the lattice in two -- making a copy of every spin along the cut. We then insert a sequence of parallelograms
which have a factor proportional to the Pauli matrix $\sigma^x$, enforcing that the spins on each side of the parallelogram are
opposite. The weight of the parallelogram is
\begin{equation}
\label{eq:spin_parallelogram}
    \includegraphics[scale=0.8, valign=c]{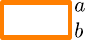} = 2^{-1/4}[\sigma^x]_{a,b} = \begin{cases}
        0 & a = b \\
        2^{-1/4} & a \neq b.
    \end{cases}
\end{equation}
With this definition, we can compute the effect of local manipulations and show that the spin defect can be moved around without
changing the partition function i.e. it is a \textit{topological defect}. In particular, for a defect to be topological, we must
show the following defect commutation relations 
\begin{equation}
\label{eq:defect_comm_1}
    \sum_{\text{internal spins}}\includegraphics[scale=0.3, valign=c]{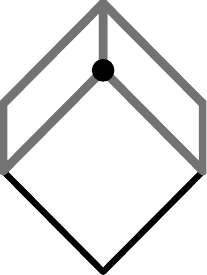} = \sum_{\text{internal
    spins}}\includegraphics[scale=0.3, valign=c, angle=180]{Figures/general_defect_rule_1.pdf},
\end{equation}
along with this diagram rotated by 90 degrees. We must also show
\begin{equation}
\label{eq:defect_comm_2}
    \sum_{\text{internal spins}} \includegraphics[scale=0.3, valign=c]{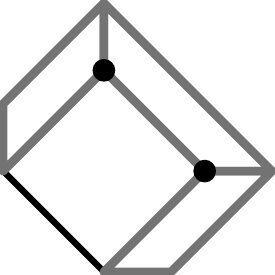} =
    \includegraphics[scale=0.3, valign=c]{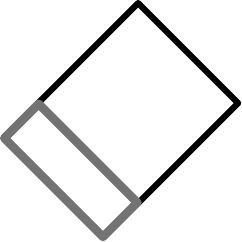},
\end{equation}
where the dots represent the weights $d_v$.
These two relations for a general defect (represented by the gray parallelograms) show that you can move the defect line around
the lattice without changing the partition function. It was shown by Aasen et al. \cite{Aasen_2016} that with the definition of
the spin defect in (\ref{eq:spin_parallelogram}), these defect commutation relations are indeed satisfied.

Clearly stacking two spin-flip defects on top of each other is equivalent to having no defect at all. For this reason it will be convenient to introduce the identity defect  which is constructed by inserting a string of parallelograms that simply identify the spins on each side
\begin{equation}
    \includegraphics[scale=0.8, valign=c]{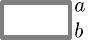}
    = 2^{-1/4}\delta_{ab}\,.
\end{equation}

\subsection{The duality defect}
The Ising model has one more non-trivial defect which implements the Kramers-Wannier duality of the Ising model \cite{Aasen_2016}.
The Kramers-Wannier duality replaces spins on the lattice with spins on the dual lattice. Hence, a duality defect must stitch
together spins on $G$  on one side with spins on $G^*$  on the other -- this is where the $\widetilde G$ picture comes into its own.
We define the parallelogram for the duality defect as
\begin{equation}
    \includegraphics[scale=0.8, valign=c]{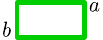} = 2^{-1/2}(-1)^{ab}.
\end{equation}
It was shown in \cite{Aasen_2016} that this too satisfies the defect commutation relations (\ref{eq:defect_comm_1} -
\ref{eq:defect_comm_2}) proving that it is a topological defect. An important difference between the spin defect and the duality
defect is the position of the spins. The spins are on diagonally opposite vertices of the parallelogram for the duality
defect. This has the effect of changing the plaquettes from vertical to horizontal and vice versa
i.e. $W^H(u_H) \to W^V(u_H)$ and $W^V(u_V) \to W^H(u_V)$.

As we already noted, the combination of two spin flip defects gives an identity defect. Further direct calculation \cite{Aasen_2016} shows that any three defect lines, labelled $\alpha,\beta,\gamma$, can interact at a triangle defect
\begin{equation}
     \includegraphics[scale=0.2, valign=c]{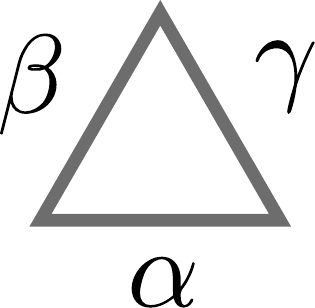} 
\label{eqn:triangle}
\end{equation}
which is non-zero only for certain allowed defect combinations. These are determined the Ising \emph{fusion category}, of which the identity, spin-flip and duality defects  form a complete microscopic representation, and to which we now turn.

\subsection{The Ising fusion category}
The microscopic calculations shown in the previous sections and computed in detail in \cite{Aasen_2016} allow us to forget about individual spins/plaquettes and instead
deal directly with the extended defect lines. They form the structure of a \textit{fusion category} that we will
describe in this section.
For completeness, we specify all the data of the Ising fusion category. We will not give a rigorous definition of fusion categories
here but simply state the objects and the formal manipulations of them that are permitted. 
Any fusion category has a finite set of objects $L$, which in the case of the Ising category are the duality field $\sigma$, the
spin field $\psi$ and the identity $\mathds{1}$. Each object is represented diagrammatically as a line in a graph. These objects satisfy an algebra, which 
in general, can be written as
\begin{equation}
  a\times b = \sum_c N^{c}_{ab}\,c,
\end{equation}
where $a,b,c \in \{\mathds{1}, \psi, \sigma\}$ and $N^{c}_{ab} \in \mathbb{N}$. The Ising algebra is
\begin{equation}
  \label{eq:fusion_alg}
    \psi \times \psi = \mathds{1};~~~\psi \times \sigma = \sigma \times \psi = \sigma;~~~\sigma \times \sigma = \mathds{1} + \psi.
\end{equation}
The vertex factors $d_a$ in (\ref{eqn:quantum_dims}) are called the \textit{Frobenius-Perron dimensions} or \textit{quantum dimensions}
and are defined as the maximal eigenvalues of the algebra coefficient matrix $[N_a]^c_b$. Again, for the Ising category these are
$d_\psi = d_\mathds{1} = 1$ and $d_\sigma = \sqrt{2}$.
The final piece of data in a fusion category is the set of \textit{F-symbols}. These dictate the  rules for manipulations,   called \textit{F-moves}, of the diagrammatic representation of the combination properties of the objects in $L$.
A general $F$-move takes the form
\begin{equation}
  \label{eqn:f-move}
  \includegraphics[scale=0.5, valign=c]{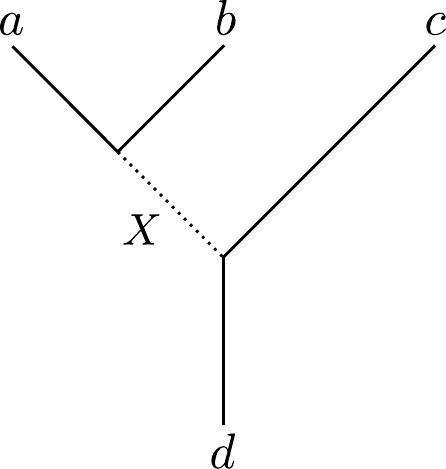}
  = \sum_{Y} [F^{abc}_d]_{XY} 
  \includegraphics[scale=0.5, valign=c]{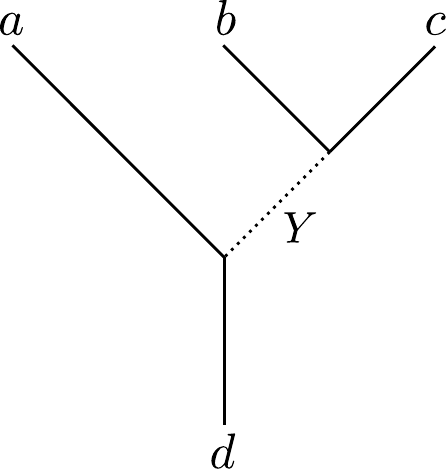}
\end{equation}
where the numbers $[F^{abc}_d]_{XY} \in \mathbb{R}$ are the $F$-symbols. These are not independent free parameters, they must
satisfy a number of self-consistency conditions \cite{aasen2020topological}.
One such solution for the Ising category has the following non-trivial $F$-symbols:
\begin{equation}
\begin{split}
  [F_\sigma^{\sigma\sigma\sigma}]_{\mathds{1}\mathds{1}} &= [F_\sigma^{\sigma\sigma\sigma}]_{\mathds{1}\psi} = [F_\sigma^{\sigma\sigma\sigma}]_{\psi\mathds{1}} = \frac{1}{\sqrt{2}}, \\
  [F_\sigma^{\sigma\sigma\sigma}]_{\psi\psi} &= -\frac{1}{\sqrt{2}}, \\
  [F_\psi^{\sigma\psi\sigma}]_{\sigma\sigma} &= [F_\sigma^{\psi\sigma\psi}]_{\sigma\sigma} = -1.
\end{split}
\label{eqn:ising_f_symbols}
\end{equation}
Any vertex that would imply the fusion of objects not allowed by the fusion algebra has a vanishing $F$-symbol. The remaining
non-zero $F$-symbols in the Ising category are 1 if allowed by the fusion algebra. 
We can compute an arbitrary planar fusion diagram using (\ref{eqn:f-move}--\ref{eqn:ising_f_symbols}) and the following rules:
\begin{equation}
    \includegraphics[scale=0.3, valign=c]{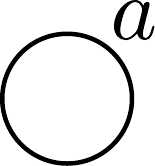} = d_a, ~~~
    \includegraphics[scale=0.3, valign=c]{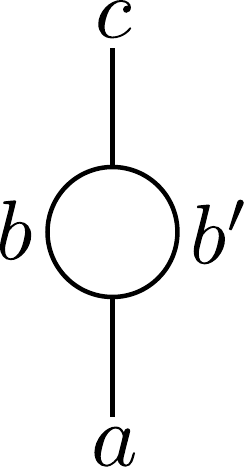} = \delta_{ac} \sqrt{\frac{d_b d_{b'}}{d_a}}
    ~\includegraphics[scale=0.3, valign=c]{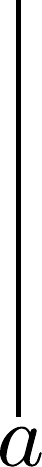}.
  \label{eq:diagram_rules}
\end{equation}
In Figure \ref{fig:f_moves} we show some examples of the application of the $F$-moves to the manipulation of topological defects.

\begin{figure}[h]
  \begin{center}
    \includegraphics[width=0.3\textwidth]{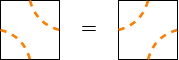}
    \[
    \includegraphics[width=0.1\textwidth, valign=c]{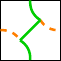}
    = - 
    \scalebox{-1}[1]{\includegraphics[width=0.1\textwidth, valign=c]{Figures/vertical_duality_horiz_spin.pdf}}
    \]
    \[
    \includegraphics[width=0.1\textwidth, valign=c]{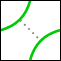}
    = \frac{1}{\sqrt{2}}\left[\,
      \includegraphics[width=0.1\textwidth, valign=c]{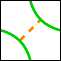}
      +
      \scalebox{-1}[1]{\includegraphics[width=0.1\textwidth, valign=c]{Figures/duality_with_identity_between.pdf}}
    \,\right]
    \]

  \caption{Non-trivial $F$-moves in the Ising fusion category. All of these moves are local and require no particular topology. The
  green solid line is the duality defect $\sigma$, the orange dashed line is the spin defect $\psi$ and the thin dotted line is
the identity $\mathds{1}$ (which is often omitted in the diagrams).}
  \label{fig:f_moves}
  \end{center}
\end{figure}

This association between topological defects in lattice models and fusion categories was fully established in \cite{aasen2020topological} where it is shown that the construction extends to   general height models.

\subsection{Dehn twist operators on the lattice}
\label{sec:conformal_dims}
A Dehn twist is an operation that cuts out a concentric circle at some finite distance from the origin, twists one half by a full revolution and glues the two halves back together at the end.
Here we show that operators that implement a Dehn twist can be constructed on a CT in a similar fashion to those constructed for a regular lattice in \cite{Aasen_2016}. 

Consider the punctured plane $C \in \mathcal C_\infty \setminus D$, where $D$ is a topological disk. 
This is a natural space to consider for a CT because we often think of the graphs as starting from the inner boundary cycle $S_0(C)$ and evolving radially, equivalently in height. Let $C^h\subset C$ denote the subgraph of $C$ consisting of all vertices and edges within a graph distance $h+1$ from $S_0(C)$ and including $S_h$ as the outer boundary.
Define $Z(C^h)$ to be the
Ising partition function on $C^h$, with some boundary conditions described below,
and represent it by the diagram
%
\begin{equation}
  Z(C^h) = \includegraphics[scale=1.1, valign=c]{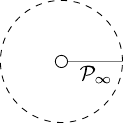}\,,
\end{equation}
where $\mathcal P_\infty$ is a path from a marked point on the inner boundary to the outer boundary; it is convenient to choose  $\mathcal P_\infty$ to be the relevant segment of the spine of $\beta^{-1} (\mathcal C)$.
Now insert an Ising topological defect that starts at a marked point on the inner boundary and extends to the outer boundary, crossing $\mathcal P_\infty$ $n$ times; again it is convenient to choose the points where the defect intersects the boundaries to lie on the spine of $\beta^{-1} (\mathcal C)$ as these are always well defined. We denote the corresponding partition function $Z_n^\phi(C^h)$ where $\phi = \mathds{1}, \psi, \sigma$ and 
represent it diagrammatically by 
\begin{equation}
  Z_n^\phi(C^h) = \includegraphics[scale=1.1, valign=c]{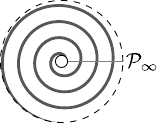}\,.
\end{equation}
The boundary conditions are chosen so that the defects cannot be moved at the boundaries and hence unravelled. 
We show in Appendix \ref{app:boundary_defects} that this can always be accomplished by an appropriate choice of spin values where the defect crosses the boundary. 

This construction can also be viewed in a transfer matrix formalism \cite{hernandez_bounds_2013}. For any height $0\le t\le h$ we define
the spin configuration $\ket{\{h_i^t\}}$ as a vector in a Hilbert space $\mathcal{H}_\phi$. The transfer matrix evolves
a state $\ket{\{h_i^t\}}$ to a state $\ket{\{h_i^{t+1}\}}$. There is a subtle distinction here: unlike a regular square
lattice, each height in a CT does not generally have the same number of spins. While this complicates writing a compact
expression for the transfer matrix, this issue is purely aesthetic. 
%
%
We now show how to construct Dehn twists in the transfer matrix formalism.
\begin{lemma}[Dehn twist operators]\label{DehnTwist}
For all $C \in \mathcal C_\infty \setminus D$, each defect $\phi = \psi, \sigma$ and every height $t$ there is an operator $\mathbf T_\phi$ acting on the Hilbert space $\mathcal H_\phi$ at height $t$
satisfying the relations
\begin{enumerate}
  \item $\mathbf T_\psi ^2 - \mathds 1_\psi = 0$,
  \item $\mathbf{T}_\sigma^4 - \sqrt{2}\mathbf{T}_\sigma^2 + \mathds{1}_\sigma = 0$,
\end{enumerate}
where the operators $\mathds 1_\phi$ act as the identity in the presence of a vertical defect $\phi$, i.e. $\mathds{1}_\phi\ket{\{h_i\}} = \ket{\{h_i\}}$.
\label{lemma:dehn}
\end{lemma}
\begin{proof}
Consider the first case, where we have a spin defect that runs from the inner boundary out to infinity without crossing
$\mathcal P_\infty$ represented by the diagram
\begin{equation}
  Z^\psi_0(C^h)= \includegraphics[scale=1.1, valign=c]{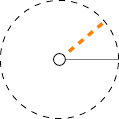}\,.
\end{equation}
Choose $ h> t'> t>0$ and define the action of operators $\mathds 1_\psi$ and $\mathbf T_\psi$  at  height $t$ by cutting open the graph 
and inserting respectively a line of identity and spin-flip defect plaquettes  as shown:
%
%
\begin{equation}
  \mathds{1}_\psi = 
    \includegraphics[valign=c, width=0.7\textwidth]{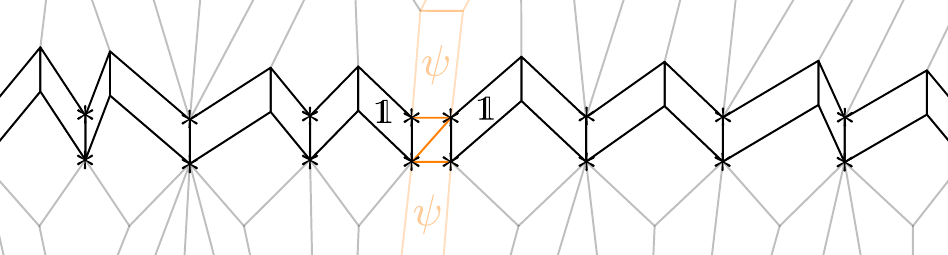}
  \label{eq:identity_spin_plaquettes}
\end{equation}
\begin{equation}
    \mathbf{T}_\psi = 
    \includegraphics[valign=c, width=0.7\textwidth]{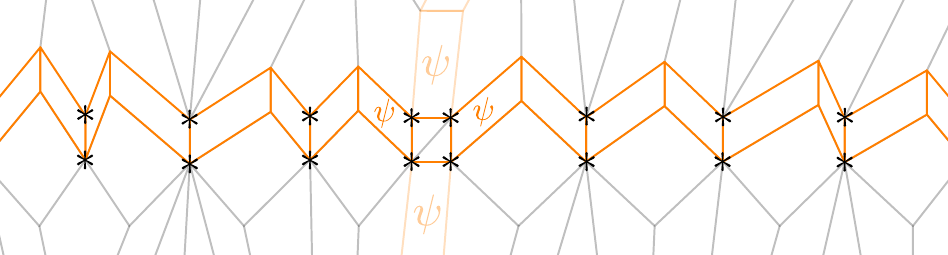}
  \label{eq:spin_spin_plaquettes}
\end{equation}
The square where the defects cross is simply two triangle defects \eqref{eqn:triangle} joined by an edge and summed over the defect label of that edge. 
%
%
Applying local moves in the partition function picture one can move the plaquettes to obtain:
\begin{equation}
  \includegraphics[valign=c, width=0.6\textwidth]{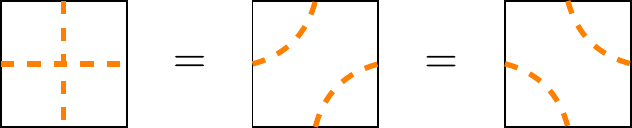}\,. 
\end{equation}
It follows that the partition function on $C^h$ with a vertical spin-flip defect of twist number $n=0$ and a horizontal spin-flip defect $\mathcal{D}_\psi$ inserted at  height $t$ satisfies
\begin{equation}
Z_0^{\psi\psi_t}(C^h)=Z_1^\psi(C^h).
\end{equation}
That is to say the action of the operator $ \mathbf{T}_\psi$ implements the Dehn twist. 
%
%
Now applying $\mathbf T_\psi^2$ at height $t$, we can use local moves to move one insertion of $\mathbf{T}_\psi$ to a
height $t'>t$ and apply $F$-moves to the defect lines to obtain
\begin{equation}
  Z_0^{\psi\psi_t\psi_t}(C^h)=Z^\psi_2(C^h) = \includegraphics[scale=1.1, valign=c]{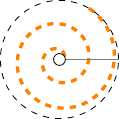} 
  = \includegraphics[scale=1.1, valign=c]{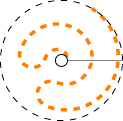} 
  = \includegraphics[scale=1.1, valign=c]{Figures/spin_twisted.pdf}
= Z^\psi_0(C^h)\,.
\end{equation}
In the transfer matrix picture we therefore have   $\mathbf{T}_\psi^2\ket{\{h_i\}} = \mathds{1}_\psi\ket{\{h_i\}}$. These manipulations apply for any $h>t+2$ so we can take $h$ to infinity thus proving part 1 of the Lemma.

The calculation for the duality defect follows the same lines but with some added technical details. Consider a single
duality defect that starts from the inner boundary and extends to infinity without wrapping around the centre shown diagrammatically by
\begin{equation}
  \label{eqn:duality_id}
  Z^\sigma_0(C^h) = \includegraphics[scale=1, valign=c]{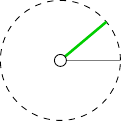}\,.
\end{equation}
There is an additional subtlety in the duality case compared to the previous spin one. The duality defect interfaces between spins
living on the vertices of $C$ and spins living on the vertices of $C^*$. Due to the topology of the space we are
considering, there must be another boundary where $C$ and $C^*$ meet again -- this is the wall. We choose to place the
wall along $\mathcal P_\infty$ for convenience. Crucially, we show in Appendix \ref{app:wall} that the wall does not pose a problem in
any manipulations that follow.

As before, we define an identity operator and an operator $\mathbf T_\sigma$ that act on the spins at height $t$ in terms of the plaquettes
\begin{equation}
  \mathds{1}_\sigma = \includegraphics[width=0.8\textwidth, valign=c]{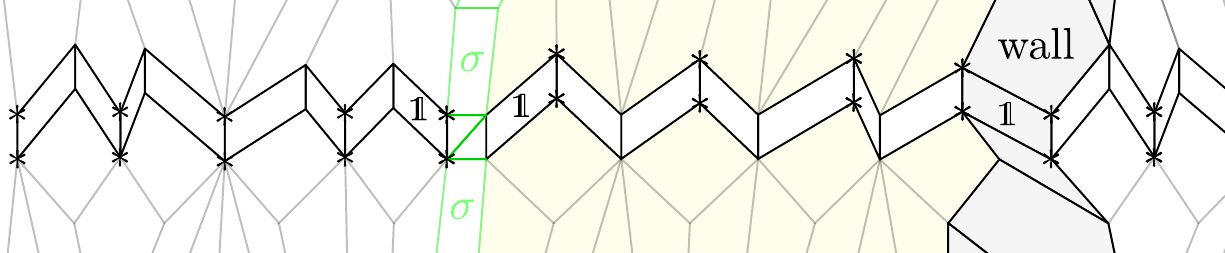},
  \label{eqn:duality_id_plaquettes}
\end{equation}
\begin{equation}
    \mathbf{T}_\sigma = 
    \includegraphics[valign=c, width=0.8\textwidth]{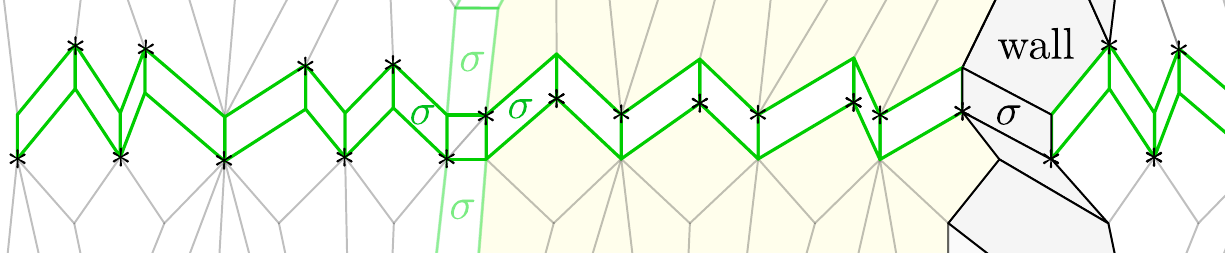},
  \label{eq:duality_duality_plaquettes}
\end{equation}
where the appropriate wall plaquettes are included. The operators $\mathds{1}_\sigma$ and $\mathbf{T}_\sigma$ can also be defined
on the slice where the spins are on the other diagonals of the duality plaquettes. This gives us two isomorphic Hilbert spaces
$\mathcal{H}_\sigma$ and $\hat{\mathcal{H}}_\sigma$. Currently, $\mathbf{T}_\sigma$ takes $\mathcal{H}_\sigma$ to $\hat{\mathcal{H}}_\sigma$ 
and vice-versa. Instead we can unify this into one operator that acts on $\mathcal{H}_\sigma \oplus \hat{\mathcal{H}}_\sigma$
\cite{Aasen_2016}.

Applying $\mathbf{T}_\sigma^2$ at height $t$ and applying local topological moves, we can use the partition function picture
and $F$-moves to compute the result:
\begin{equation*}
Z^\sigma_2(C^h) = 
\includegraphics[scale=1, valign=c]{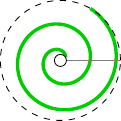}
= \frac{1}{\sqrt{2}}
\left[
\includegraphics[scale=1, valign=c]{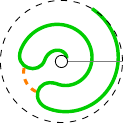}
+
\includegraphics[scale=1, valign=c]{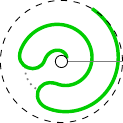}
\right],
\end{equation*}
\begin{equation*}
= \frac{1}{\sqrt{2}}\left[
  \includegraphics[scale=1, valign=c]{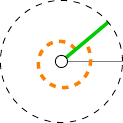}
+
  \includegraphics[scale=1, valign=c]{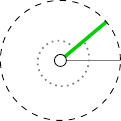}
\right],
\end{equation*}
where the details of how the defects pass through the wall can be found in Appendix \ref{app:wall}.
In terms of the operators:
\begin{equation}
    \mathbf{T}_\sigma^2 = \frac{1}{\sqrt{2}}\left(\psi_\sigma + \mathds{1}_\sigma\right),
\end{equation}
\begin{equation}
    \mathbf{T}_\sigma^4 = \frac{1}{2}\left(\psi_\sigma + \mathds{1}_\sigma\right)^2 = \frac{1}{2}\left(\mathds{1}_\sigma +
    2\psi_\sigma + \psi_\sigma^2 \right),
\end{equation}
where $\psi_\sigma$ is an operator made of horizontal spin plaquettes in the presence of a vertical duality defect -- we omit an explicit diagram for the sake of brevity.
It remains to calculate $\psi_\sigma^2$, which is most easily done diagrammatically with the use of $F$-moves:
\begin{equation*}
  \includegraphics[scale=1, valign=c]{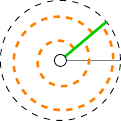}
=
\includegraphics[scale=1, valign=c]{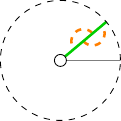}
= -~
\includegraphics[scale=1, valign=c]{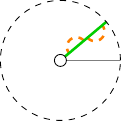}
=-~
\includegraphics[scale=1, valign=c]{Figures/duality_identity.pdf}
= -Z_0^\sigma(C^h),
\end{equation*}
or in terms of operators $\psi_\sigma^2 = - \mathds 1_\sigma$,
which gives us an algebraic equation for $\mathbf{T}_\sigma$:
\begin{equation}
    \mathbf{T}_\sigma^4 = \sqrt{2}\mathbf{T}_\sigma^2 - \mathds{1}_\sigma.
\end{equation}
As before, this applies for any $h>t+2$ so we can take $h$ to infinity thus proving part 2 of the Lemma.
\end{proof}

\begin{corollary}[Eigenvalues of Dehn twist operators]\label{DTEvals}
  The eigenvalues of the operators $\mathbf T_\psi$ and $\mathbf T_\sigma$ are given by
  \begin{itemize}
    \item $\mathbf T_\psi: \lambda = \pm 1$
    \item $\mathbf T_\sigma: \lambda = e^{2\pi i/16},~e^{-2\pi i/16},~e^{-2\pi i\cdot 7/16},~e^{2\pi i \cdot 7/16}$
  \end{itemize}
  \label{cor:dehn}
\end{corollary}
\begin{proof}
  Both relations follow directly by turning the operator equations in Lemma \ref{lemma:dehn} into eigenvalue equations. 
In the spin defect case, this becomes $\lambda^2 = 1$.
In the duality defect case, we have $\lambda^4 - \sqrt{2}\lambda^2 + 1 = 0$.
\end{proof}

The Dehn twist operators described here are a topological construction on the graph $C$. Lemma \ref{DehnTwist} and Corollary \ref{DTEvals} are exact statements of their properties which are independent of the values of the Ising weights and will persist in any continuum limit. Provided that the Dehn twist itself exists on continuum CTs, we can equate  the  eigenvalues of $\mathbf T_{\psi,\sigma}$ to those of $e^{2\pi i (L_0 - \bar{L}_0)}$ which is the known form of the operator in the continuum. Before we can do this, we must first prove that it is possible to construct a Dehn twist on continuum CTs, which we do in the following section.


\section{Stochastic formulation of CDT}
\label{sec:stochastic}
We derive a stochastic differential equation for the length of the spatial slice at time $t$ starting from the discrete
Galton-Watson process as an alternative to the results in \cite{lamperti_ney1968, sisko2011note}. We then show how this generates
a random measure on some base space in an exact analogy to the Liouville gravity picture.
\subsection{The \lpfull}
\begin{theorem}[\lpfull]
\label{thm:sde}
Let \(L(t)\) be the length of the spatial slice at time t. Then \(L(t)\) satisfies the It\^o integral equation
\begin{equation}
    L(t) = L(0) + \int_0^t \sqrt{f''(1)L(s)}\,dW(s) + \int_0^t f''(1)\,ds,
\end{equation}
or equivalently, the stochastic differential equation
\begin{equation}
    dL(t) = \sqrt{f''(1) L(t)}dW(t) + f''(1)dt,
\end{equation}
where $f(s)$ is the generating function of the GW process.
We call L(t) a \lpfull~(\lp) since it was first discovered by them \cite{lamperti_ney1968}.
\end{theorem}
Before we prove this we need to introduce some definitions and prove some lemmas.
We showed that the ensemble of infinite CTs is in bijection to critical GW trees conditioned to survive forever.
The number of points at height $k+1$, $\eta_{k+1}$ is given by
\begin{equation}
\label{discrete_diff_Z}
\eta_{k+1} = \eta_k + \sum_{j=1}^{\eta_k-1}Z_j + Z_0,
\end{equation}
where $Z_j = Y_j - 1$ and $\{Y_j\}$ are i.i.d. random variables whose distribution is given by the generating function $f(s)$,
i.e. they are the random variables whose value is the number of offspring of a point on the spatial slice. Similarly, $Z_0 = Y_0 -
1$ where $Y_0$ is the number of offspring of the special vertex on the infinite spine, which is distributed according to $sf'(s)$
\cite{durhuus2022trees}.
We list the following results for future reference:
\begin{equation}
\begin{split}
    \mathbb{E}Z_k &= 0, \\
    \mathbb{E}Z_0 &= f''(1), \\
    \mathbb{E}[Z_k^2] &= f''(1).
\end{split}
\end{equation}
We define the re-scaled process $L_k^n := \eta_k/n$, where $n$ is some integer which we will take to infinity at the end.
Rewriting (\ref{discrete_diff_Z}) in terms of $L_k^n$ becomes 
\begin{equation}
\label{eqn:scaled_diff}
L_{k+1}^n = L_k^n + \frac1n \left(\sum_{j=1}^{nL_k^n-1}Z_j + Z_0\right)
\end{equation}

\begin{lemma}
\label{lemma:martingale}
Define
\begin{equation}
  \xi_{k+1}^n = \frac{1}{\sqrt{nL_k^n}}\left[\sum_{j=1}^{nL_k^n - 1} Z_j + Z_0 - \mathbb{E}Z_0 \right].
\end{equation}
Then $W_n(t) := \frac{1}{\sqrt{n}}\sum_{k=1}^{[nt]}\xi_k^n$ is a Martingale and $W_n(t) \Rightarrow W(f''(1)t) = \sqrt{f''(1)}W(t)$ where
$W$ is a standard Brownian motion. The convergence $(\Rightarrow)$ is in distribution.
\end{lemma}
\begin{proof}
  From Ethier \& Kurtz \cite{ethier_kurtz_book} $W_n(t)$ is a Martingale if $\mathbb{E}[\xi_{k+1}^n|\mathcal{F}_k^n] = 0$, where
  $\mathcal{F}_k^n$ is the filtration at $k$. This is true since given $L_k$, $\xi_{k+1}^n$ is a sum of random variables
  with mean 0. To show convergence in distribution as it is sufficient to show
  \[
    \frac1n \sum_{k=1}^{[nt]}(\xi_k^n)^2 \to f''(1)t
  \]
  in probability as $n\to\infty$.
  \begin{equation}
  \begin{split}
  \frac1n \sum_{k=1}^{[nt]}(\xi_k^n)^2 &= \frac1n \sum_{k=1}^{[nt]}\frac{1}{nL_k}\left(\sum_{i=1}^{nL_k-1}Z_i + Z_0 - \mathbb{E}[Z_0]\right)^2 \\
                                       &= \frac1n \sum_{k=1}^{[nt]}\frac{1}{nL_k}\sum_{i=0}^{nL_k-1}Z_i^2 + ... \\
                                       &\to f''(1)t
  \end{split}
  \end{equation}
  where in the second to last line, we ignore cross terms (because the $Z_i$ are independent) and the terms which 
  will surely vanish as $n\to\infty$. In the final line, we used the law of large numbers and $\mathbb{E}[Z_i^2] = f''(1)$.
\end{proof}

\begin{proof}[Proof of Theorem \ref{thm:sde}]
We can re-write equation \ref{eqn:scaled_diff} in terms of $\xi_k$.
\begin{equation}
  L_{k+1}^n = L_k^n + \frac{\sqrt{L_k^n}}{\sqrt{n}}\xi_{k+1}^n + \frac{\mathbb{E}[Z_0]}{n}
\end{equation}
Following Kurtz and Protter \cite{kurtz_protter1991} we define $L_n(t) := L^n_{[nt]}$, $W_n(t) :=
\frac{1}{\sqrt{n}}\sum_{k=1}^{[nt]}\xi_k^n$ and $V_n(t) := [nt]/n$ where $[nt]$ denotes the nearest integer to $nt$. It then follows
that 
\begin{equation}
    L_n(t) = L_n(0) + \int_0^t \sqrt{L_n(s)}dW_n(s) + \int_0^t\mathbb{E}[Z_0]~dV_n(s).
\end{equation}
By Lemma \ref{lemma:martingale} we know $W_n(t) \Rightarrow \sqrt{f''(1)}W(t)$ where $W$ is a standard Wiener process and
$V_n(t) \Rightarrow V(t) = t$. Kurtz and Protter \cite{kurtz_protter1991} showed in general that for any process of this form, $L_n(t) \Rightarrow
L(t)$ where $L(t)$ satisfies the integral stochastic equation
\begin{equation}
  L(t) = L(0) + \int_0^t \sqrt{f''(1)L(s)}dW(s) + \int_0^t \mathbb{E}[Z_0]~ds
\end{equation}
or in differential form
\begin{equation}
\label{eq:the_sde}
    dL(t) = \sqrt{f''(1) L(t)}dW(t) + f''(1)dt
\end{equation}
where we have used the fact that $\mathbb{E}[Z_0] = f''(1)$.
\end{proof}
\noindent
Re-scaling $L(t) \rightarrow 2L(t)/f''(1)$ gives us exactly the result found in  \cite{sisko2011note}. 

\subsection{Scaling dimensions of Ising CFT fields}
\label{sec:dehn}
Equipped with the existence of the \lpfull, we can now explicitly construct a Dehn twist in a randomly sampled continuum CT as follows. 
Identify the ends of the spatial slices to obtain a space with the topology of a
cylinder, where the circumference at height $t$ is $L(t)$.  
Choose a curve $c$ at height $t_0+1/2$ and let $A$ be a
neighbourhood of $c$ which is homeomorphic to $S^1 \times [t_0, t_0+1]$. Noting that $L(t)$ is a continuous function, we define coordinates $(s,t)$ on $A$ where $s = e^{2\pi i
x/L(t)}$, with $x \in [0, L(t)]$ and $t \in [t_0, t_0+1]$. Then the Dehn twist $f$ is defined as 
\[
  f:(e^{2\pi i x /L(t)}, t) \mapsto (e^{2\pi i (x + (t-t_0)L(t))/L(t)}, t)
\]
or, in terms of the $x$ coordinate, $x \mapsto x + (t-t_0)L(t)$. We observe that this construction does not exist for the 
Liouville case as the stochastic process described in section \ref{sec:StochasticLiouville} is not continuous.

\begin{theorem}[Scaling dimensions of Ising CFT operators]
  Assume that a continuum limit of the Ising model on a CT exists, then there exist fields in the continuum CFT with spin ($L_0 - \Bar{L}_0$) given by
  \begin{itemize}
    \item $h_\psi = 1/2 + n\,, n \in \mathbb Z$
    \item $h_\sigma = 1/16 + n\,, n \in \mathbb Z$
  \end{itemize}
\label{lem:ScalingDims}
\end{theorem}
\begin{proof}
  Firstly, the results of Corollary \ref{cor:dehn} are topological and so given the existence of a continuum limit, they persist at all scales. In the continuum, the known form of the Dehn twist operator is given by $e^{2\pi i (L_0 - \bar{L}_0)}$ and we have shown explicitly that it exists. Hence we can read off the spin of the field associated with $\psi$ as $h_\psi = 1/2 + n,\; n\in\mathbb Z$, where we have chosen the $\lambda = -1$ sector. We are free to choose the sector because local manipulations of plaquettes prove that the Dehn twist operators $\mathbf T_\psi$ and $\mathbf T_\sigma$ always commute with the transfer matrix. Therefore, we can label a state by the eigenvalues $\lambda_\psi$, $\lambda_\sigma$ and the spin configuration. Similarly for the duality defect, comparing the results of Corollary \ref{cor:dehn} to the continuum Dehn twist operator, we find that there is a sector where $h_\sigma = 1/16 + n,\;n\in\mathbb Z$.
\end{proof}

\subsection{Properties of the Lamperti-Ney process}
\label{sec:lnp}
As we showed above, the \lp~is defined by 
\AMSeqn{dL(t)=\sqrt{f''(1)L(t)}\,dW(t)+f''(1)\,dt\label{eqn:LPdefn}}
where $W(t)$ is a standard Wiener Process (WP) which has the property that
\AMSeqn{W_1(t)&=x^{-\half}W(xt)\label{eqn:WPpropA}\\
W_2(t)&=W(t+x)-W(x)\\
W_3(t)&=tW(-t^{-1})
}
are all WPs. In general if $g\in SL(2,\mathbb R)/\{\pm \mathds{1}\}$ (with the usual representation and $ad-bc=1$) then
\AMSeqn{W_g(t)=(ct+d) W\left(\frac{at+b}{ct+d}\right)-ct\,W\left(\frac{a}{c}\right)-d\,W\left(\frac{b}{d}\right)}
is also a WP.

\begin{lemma} \label{lem:LPequiv} If $L(t)$ is an \lp, then so is $L_x(t)\equiv x^{-1}L(xt)$.
\end{lemma}
\begin{proof}
From \eqref{eqn:LPdefn} we have
\AMSeqn{dL_x(t)=d(x^{-1}L(xt) )&= x^{-1} \,dL(xt)\nn\\
&=\sqrt{f''(1)\, x^{-1}L(xt)}\,x^{-\half} dW(xt)+f''(1)\,dt\nn\\
&=\sqrt{f''(1)\, L_x(t)}\,dW_1(t)+f''(1)\,dt\,, }
where we have used \eqref{eqn:WPpropA}.
\end{proof}
\noindent
Now consider the `square'  segment of an LNP 
shown in figure \ref{fig:LPball},
\begin{figure}
\begin{center}
  \includegraphics[scale=0.8]{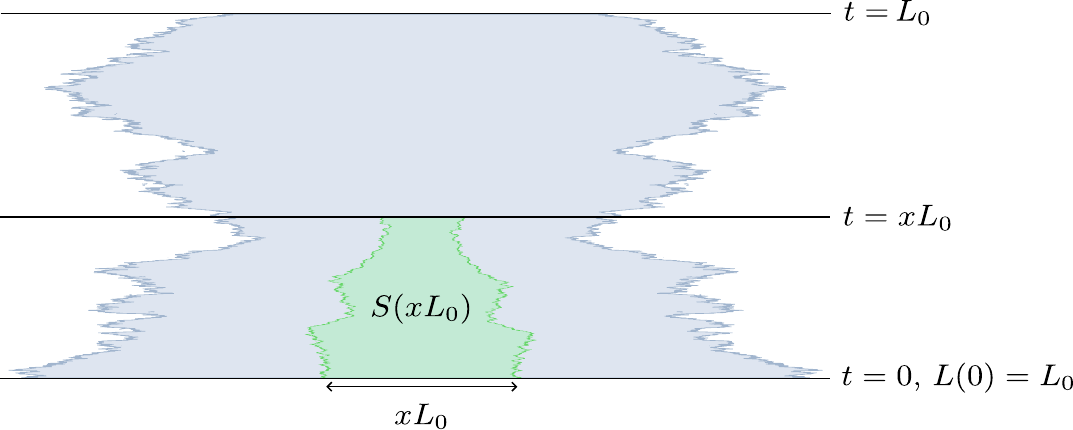}
  \end{center}
\caption{A segment of an LNP. The region shaded green is generated from a subset of the process at $t=0$ and has area $S(xL_0)$.}
\label{fig:LPball}
\end{figure}
then
\begin{lemma} The square area $S_x(L_0) \equiv x^{-2} S(xL_0)$, where $x>0$, is equal in law to  $S(L_0) = \int_{s=0,L(0) = L_0}^{L_0}L(s)\,ds$.
  \label{lem:area}
\end{lemma}
\begin{proof}
Consider the scaled area process 
\AMSeqn{S_x(L_0) = x^{-2}S(xL_0) = x^{-2}\int_{s=0, L(0) = xL_0}^{xL_0} L(s)\,ds.}
Let $s = xt$, then
\AMSeqn{S_x(L_0) = \int_{s=0,\, x^{-1}L(0) = L_0}^{L_0} x^{-1}L(xt)\,dt.}
Applying Lemma \ref{lem:LPequiv} to the r.h.s. then gives
\AMSeqn{S_x(L_0) = \int_{s=0, L_x(0) = L_0}^{L_0} L_x(t)\,dt,}
which is equal in law to $S(L_0)$. 
\end{proof}
\noindent
In other words, $S(xL_0)$ is equal in law to $x^2S(L_0)$ and so scales like a canonical 2D area.

\begin{lemma}
\label{lem:LPositive}
The process $L(t), L(0) = L_0 > 0$ is strictly positive at all positive times.
\end{lemma}
\begin{proof}
Let $L(t) = f''(1) L'(t)/4$, then 
$$
dL'(t) = 4dt + 2\sqrt{L'(t)}dW(t).
$$
This is an $n=4$ squared Bessel process which is strictly positive for all positive $t$ \cite{RevuzMartingales}.\footnote{The squared Bessel process governs the distance from the origin of a Brownian walk in $d=n+1$ dimensions that starts at the origin. The process $L(t)$ thus describes the segment of such a walk from the time that it first reaches $L_0$. Brownian walks in $d>2$ dimensions are non-recurrent so the walk never revisits the origin and $L(t)$ is strictly positive.}
\end{proof}

\subsection{Classical and quantum scaling exponents in CDT}
In this section we give another proof, in the spirit of the Duplantier-Sheffield construction, that scaling exponents on causal random geometry do not shift according to a KPZ-like relation. 
In the same way that DS used the volume term in the Liouville action to define a random measure $d\mu
= e^{\gamma\phi}d^2z$ we can define a random measure for continuum CDTs as $d\mu = L(t)dtdx$, where $L(t)$ is a \lpfull 
~(\ref{eq:the_sde}). The spatial direction is uniform so we may write the two-dimensional measure as $d\mu =
L(t)dt$, since the base space is $[0,1]\times[0,1]$. It is clear that (\ref{eq:the_sde}) is Lipschitz continuous on the domain $(0,\infty)$ and so has a continuous strong solution.
The solution $L(t)$ is a proper function of $t$ so no regularization of distributions is required in contrast to the case with the Liouville measure.

We are now in a position to show that the KPZ formula does not apply to CDTs in the continuum and in fact there is no shift
in the scaling dimensions of fields on CDTs compared to a fixed lattice (up to logarithmic corrections). It will be convenient to rewrite the definitions of the scaling exponents (\ref{eq:classical_exponent}--\ref{eq:quantum_exponent}) in a discrete form. First, consider dividing the base space $[0,1]\times[0,1]$ in two ways:
\begin{enumerate}
\item Into base cells $\mathcal B = \{B_\alpha,\,\alpha = 1,...,1/\varepsilon^2\}$ of classical area $\varepsilon^2$, $\varepsilon \ll 1$. 
\item Into quantum cells $\{C_\alpha,\,\alpha=1,...,S(1)/\delta \}$ of quantum area $\delta$, i.e. such that $\mu(C_\alpha) = \delta$. Let $\mathcal B_\alpha$ denote the set of base cells in $C_\alpha$ and $K_\alpha = |\mathcal B_\alpha|$.
\end{enumerate}
\begin{figure}
  \begin{center}
    \includegraphics[width=0.6\textwidth]{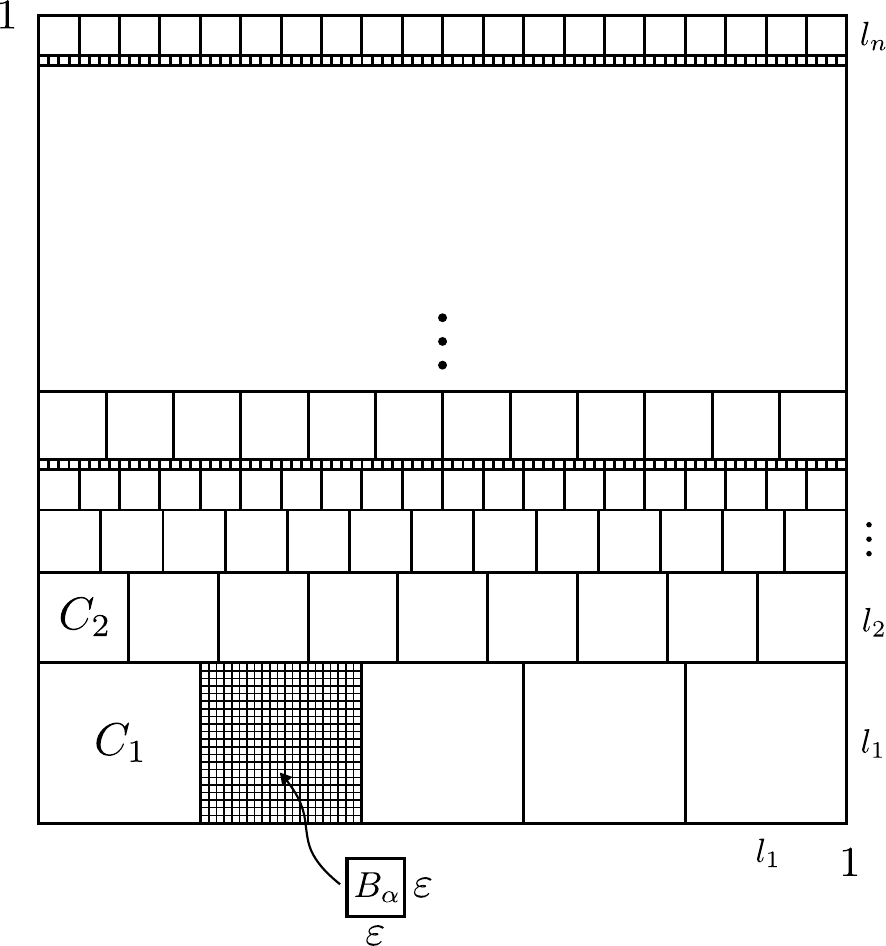}
  \end{center}
  \caption{A decomposition of the base space $[0,1]\times[0,1]$ according to the measure $\mu$.}\label{fig:cdt_dyadic}
\end{figure}

Note that the decomposition according to the measure $\mu$ is much simpler than that of Liouville gravity. Due to the one--dimensional nature, the base space is divided into rows $i = 1, ..., n$, where each quantum cell in a row is the same size in the base space and has the same quantum area -- see figure \ref{fig:cdt_dyadic}. We use Greek subscripts when referring to an element of the complete set of cells and Latin subscripts to label the height of a cell. For example, $C_i$ is a quantum cell at a height $i$ in the base space -- there is no need to distinguish the cell in the row, since they are all copies of each other.

\newpage
Let $X$ denote a random subset of $[0,1]\times[0,1]$ and $\mathcal X = \{b \in \mathcal B: b\,\cap\, X \neq \emptyset\}$ denote the set of base cells that intersect $X$.
\begin{definition}[Discrete Euclidean and quantum scaling exponents] ~
\begin{itemize}
    \item The \textbf{Euclidean scaling exponent} $x = x(X)$ is defined as
        \begin{equation}
          \label{eq:c_exp_disc}
            x(X) := \lim_{\varepsilon \to 0} \frac{\log \mathbb{E}_X[\varepsilon^2 N(\varepsilon, X)]}{\log \varepsilon^2}
        \end{equation}
        where $N(\varepsilon, X)$ is the number of base cells that intersect X.
    \item The \textbf{quantum scaling exponent} $\Delta = \Delta(X)$ is defined as
        \begin{equation}
            \label{eq:q_exp_disc}
            \Delta(X) := \lim_{\delta \to 0} \frac{\log \mathbb{E}_X[\delta N_Q(\delta, X)]}{\log \delta}
        \end{equation}
        where $N_Q(\delta, X)$ is the number of quantum cells that intersect X.
\end{itemize}
The expectation in both cases is over the ensemble of random subsets $X$.
\end{definition}
\begin{theorem}
The quantum scaling dimension $\Delta(X)$ and the classical scaling dimension $x(X)$ are equal for the measure $d\mu = L(t)dt$,
where $L(t)$ is a \lp.
\end{theorem}

\begin{proof}
By Lemma 4.5, a quantum cell $C_i$ (see figure 8) at height $i$ and of quantum area $\delta$ covers a square in the base space of side length $\ell_i$, where $\delta = \ell_i^2 S_i$ and $S_i$ is sampled from the law $S{(1)}$. Define  the
integer $n(\delta)$  by
\begin{equation} n(\delta)=\max \{n: \sum_{i=1}^{n} l_i < 1\} \end{equation}
As $S_i$ is strictly positive and $\mathbb{E}[S_i^{-\frac{1}{2}}]$ is finite, this implies that, as $\delta \to 0$, we have $n(\delta)^2\delta = \mathcal{O}(1)$.

Let $p_X(\varepsilon) = \mathbb P[B_\alpha \cap X \neq \emptyset]$ be the \textit{a priori} probability that $X$ intersects a given base cell $B_\alpha$.
The probability $\mathbb{P}(C_i \cap X \neq \emptyset)$ that a given quarter cell $C_i$ intersects the subset $X$   is  given by
\begin{align}
  \mathbb{P}[C_i \cap X \neq \emptyset] &= \mathbb{P}[\mathcal B_i \cap \mathcal X \neq \emptyset],\nonumber \\
                                        &= 1 -\mathbb{P}[\mathcal B_i \cap \mathcal X  = \emptyset]\,.
\end{align}
Letting
\begin{align}  K_i^+=\left\lceil\frac{l_i}{\varepsilon}\right\rceil\,,\;\mathrm{and}\;K_i^-=\left\lfloor\frac{l_i}{\varepsilon}\right\rfloor\,,
\end{align}
we have 
\begin{equation}
1 - (1 - p_X(\varepsilon))^{K_i^-} \leq \mathbb{P}(C_i \cap X \neq \emptyset) \leq 1 - (1 - p_X(\varepsilon))^{K_i^+}\,.\label{eqn:A4}
 \end{equation}
The quantity $l_i^{-1}$ is not generally an integer, but we note that the quantum area, $A_i$, of the cells $C_i$ satisfies
\begin{equation}
\delta\left\lfloor\frac{1}{l_i}\right\rfloor \le A_i <  \delta\left\lceil\frac{1}{l_i}\right\rceil\,.\label{eqn:A5}
 \end{equation}
Then, using \eqref{eqn:A4} and \eqref{eqn:A5}, we find that the expectation of the area of quantum cells that intersect $X$ is bounded above by
\begin{equation}
\mathbb{E}_X\delta N_Q(\delta, X)< \sum_{i=1}^{n(\delta)+1} \delta\left\lceil\frac{1}{l_i}\right\rceil(1 - (1 - p_X(\varepsilon))^{K_i^+}),
 \end{equation}
 and below by
\begin{equation}
 \sum_{i=1}^{n(\delta)} \delta\left\lfloor\frac{1}{l_i}\right\rfloor(1 - (1 - p_X(\varepsilon))^{K_i^-})<\mathbb{E}_X\delta N_Q(\delta, X)\,.
 \end{equation}
Finally note that
\begin{equation}
p_X(\varepsilon) = \varepsilon^2 N(\varepsilon,X)\,, \end{equation}
and set $\delta = K  \varepsilon^2$, where $K\gg1$ is a fixed number. Then the quantum dimension converges to 
\begin{align}
 \Delta(X) &= \lim_{\varepsilon \to 0} \frac{\log\mathbb{E}_X\left[\sum_{i=1}^{n(K\varepsilon^2)} \frac{K\varepsilon^2}{l_i}\left(\frac{l_i^2}{\varepsilon^2}+O\left( \frac{l_i}{\varepsilon}\right)\right) \varepsilon^2 N(\varepsilon,X)\right]} {\log K\varepsilon^2}\nonumber\\
 		&=\lim_{\varepsilon \to 0} \frac{\log(K+O\sqrt{K})+ \log \mathbb{E}_X[\varepsilon^2 N(\varepsilon,X)]}{\log K\varepsilon^2}\nonumber\\
		&=x(X)\,,
\end{align}
where we have used the fact that
\begin{align}
 \lim_{\varepsilon \to 0} \sum_i^{n(K\varepsilon^2)} l_i = 1.
\end{align}
\end{proof}

\noindent
We remark that this result is an inevitable outcome of the continuous, one dimensional nature of the measure.

\section{Connection to Ho\v rava-Lifshitz gravity}
\label{sec:horava}
We note another interesting connection to projectable Ho\v rava-Lifshitz (HL) gravity. The action for
projectable HL gravity can be reduced to a one-dimensional action of the form \cite{Ambj_rn_2013}
\begin{equation}
\label{eq:hl_action}
    S_E = \int dt \left( \frac{\dot{L}(t)^2}{4L(t)} + \Lambda L(t)\right).
\end{equation}
We will show that this is related to the \lpfull.

The Onsager-Machlup (OA) function allows us to write down a Lagrangian associated to any SDE \cite{onsager_machlup, bach_hors1975, Weber_2017}. In general, for an It\^o process
\begin{equation}
    dX_t = f(X_t)dt + \sigma(X_t)dW_t
\end{equation}
the OA function is given by 
\begin{equation}
    \mathcal{L}(\dot{x}, x) = \frac{(\dot{x} - f(x))^2}{2\sigma(x)^2}.
\end{equation}
For the action \eqref{eq:hl_action}, the OA function can be read off: 
\begin{equation}
\label{eq:oa_hl}
    \mathcal{L}(\dot{L},L) = \frac{\dot{L}^2}{4L}.
\end{equation}
This is associated with the stochastic process satisfying the SDE:
\begin{equation}
    \label{eq:sde_no_drift}
    dL_t = \sqrt{2 L_t}\,dW_t.
\end{equation}
The associated formal path integral to the OA Lagrangian is
\begin{equation}
    Z = \int \mathcal{D}L e^{-\int \mathcal{L}(\dot{L},L) dt}.
\end{equation}
As is standard procedure in the probabilistic approach to quantum field theory (see \cite{Simon_1975}), we interpret the term 
\begin{equation}
\label{eq:sde_measure}
    Z^{-1}e^{-\int \mathcal{L}(\dot{L},L) dt}\mathcal{D}L
\end{equation}
as the measure which properly weights the paths given by the stochastic process. 

The process in \eqref{eq:sde_no_drift} is the \lpfull~without the constant drift term. 
We now show that one can change measure to include the drift term and that the Radon-Nikodym derivative has a very simple and suggestive form.

\begin{theorem}
\label{thm:change_measure}
Let $\mathbb{P}$ be the measure associated with the process defined by \eqref{eq:sde_no_drift} and let $\widetilde{\mathbb{P}}$ be a new measure defined by the Radon-Nikodym derivative 
\[
\frac{d\widetilde{\mathbb{P}}(t)}{d\mathbb{P}(t)} = \frac{L(t)}{L(0)}
\]
then the process $L_t$ satisfies
\[
dL_t = 2 dt + \sqrt{2 L_t}\,d\widetilde{W}_t,
\]
where 
\[
\widetilde{W}_t = W_t - \int_0^t \frac{2}{\sqrt{2L_t}}ds,
\]
is standard Brownian motion under $\widetilde{\mathbb{P}}$.
\end{theorem}

\begin{proof}
The proof is a standard application of Girsanov's theorem. We first write the process \eqref{eq:sde_no_drift} as 
\[
dL_t = 2dt + \sqrt{2L_t}\left(dW_t - \frac{2}{\sqrt{2L_t}}dt\right).
\]
Girsanov's theorem states that the process 
\[
\widetilde{W}_t = W_t - \int_0^t \frac{2}{\sqrt{2L_t}}ds,
\]
will be standard Brownian motion under the measure $\widetilde{\mathbb{P}}$ defined by the Radon-Nikodym derivative 
\[
\frac{d\widetilde{\mathbb{P}}(t)}{d\mathbb{P}(t)} = \exp\left\{\int_0^t\frac{2}{\sqrt{2L(s)}}dW(s) - 
\int_0^t \frac{1}{L(s)}ds \right\}.
\]
A straightforward application of It\^o's lemma shows that
\[
d(\log L(t)) = \frac{2}{\sqrt{2L(s)}}dW(s) -  \frac{1}{L(s)}ds,
\]
hence 
\[
\log L(t) - \log L(0) = \int_0^t\frac{2}{\sqrt{2L(s)}}dW(s) - 
\int_0^t \frac{1}{L(s)}ds,
\]
which proves the result.
\end{proof}

Under the change of measure defined in Theorem \ref{thm:change_measure}, expectations of random variables in the measure $\widetilde{\mathbb{P}}$ are given by $\widetilde{\mathbb{E}}[X] = \mathbb{E}[\frac{L(t)}{L(0)}X]$.

We are often interested in adding potential terms, such as a cosmological constant term,
to the Lagrangian. These are then considered as observables whose expectation
is being calculated with respect to the measure (\ref{eq:sde_measure}). For example, the partition function including a cosmological constant term is given by
\begin{equation}
    \int \mathcal{D}L\,e^{-\int \mathcal{L}(\dot{L(t)},L(t)) + \Lambda L(t)\,dt} = \mathbb{E}[e^{-\Lambda \int L(t)\,dt}]\,,
\end{equation}
where $\mathbb{E}$ is over paths sampled according to an SDE. For the \lp,~this is indeed the correct cosmological constant term since $\int L(t)\,dt$ is the two dimensional volume.
The expectation value of a general observable $F(L)$ is given by 
\begin{equation}
    \label{eq:expectation}
    \int \mathcal{D}L\,F(L(t)) e^{-\int \mathcal{L}(\dot{L},L)dt} = \mathbb{E}[F(L(t))]\,
\end{equation}
which may remind the reader of the expression that
satisfies the 
Feynman-Kac formula. Indeed, if we specify an initial condition, say $L_0 = l$, then $\phi(l,\tau) = \widetilde{\mathbb{E}}[\,G(L(t))e^{-\Lambda
\int_0^\tau L(t)\,dt}\, |\, L_0 = l\,]$ satisfies 
the differential equation 
\begin{equation}
  \label{eqn:feyn_kac_hamiltonian}
    -\frac{\partial\phi}{\partial\tau} = -l\frac{\partial^2\phi}{\partial l^2} - 2\frac{\partial\phi}{\partial l} + \Lambda l\phi;
    ~~~ \phi(l,0) = G(l),
\end{equation}
where the expectation is taken with respect to the measure $\widetilde{\mathbb{P}}$.
This is exactly the imaginary time Schr\"odinger equation with the CDT Hamiltonian \cite{Ambj_rn_1998,sisko2011note}. However, as discussed in \cite{Ambj_rn_2013}, there is an ambiguity in the operator ordering of the Hamiltonian when it is quantized. 
The Hamiltonian is given by $H = L\Pi^2 +
\Lambda L$, where $\Pi = -i \frac{\partial}{\partial L}$. The Hamiltonian in \eqref{eqn:feyn_kac_hamiltonian} corresponds to the ordering $H = -\frac{\partial^2}{\partial L^2}L + \Lambda L$, which is Hermitian on the space of square integrable functions with the measure $L\,dL$. This factor of $L$ has its origin as the Radon-Nikodym derivative of changing from a process $L_t$ without drift to one with a drift term $2\,dt$.
It also corresponds to the discrete model where both entrance and exit loops are unmarked \cite{Ambj_rn_2013}. 

These results further corroborate the findings that 2D continuum CDT is 2D Projectable Ho\v
rava-Lifshitz gravity, this time showing the correspondence from the other direction. 

\section{Outlook and extension to the annealed case}
\label{sec:annealed}
In the quenched model, where a graph is first sampled from the UICT ensemble, and matter fields are subsequently placed on it, we have demonstrated that there is no change in the critical exponents compared to the flat lattice for any height model.
In the particular case of the Ising model, Theorem \ref{lem:ScalingDims} corroborates the numerical results in \cite{Ambj_rn_1999, Ambj_rn_2009, Benedetti_2007}.
\textit{If} a critical point exists, the topological defects and their associated fusion algebra persist in the continuum, being independent of any coupling constant. The continuity of the stochastic process $L(t)$, guarantees the existence of the Dehn twist in the continuum. Hence, the discrete arguments apply, showing that the conformal dimensions of the Ising operators are unchanged from their Onsager values.

This is in contrast to the KPZ relation in Liouville gravity, where even in the pure gravity case ($\gamma = \sqrt{8/3}$) we get a non-trivial relation between $x(X)$ and $\Delta(X)$.
One way to explain this difference between causal and Euclidean random geometry is that there is no clear way to define a Dehn twist in the Euclidean setting. In the discrete Euclidean picture, the set of vertices at a constant geodesic distance from a chosen origin is almost surely disconnected. The Dehn twist cannot be defined because
there is no curve around which there exists a region homeomorphic to $S^1 \times [0,1]$, due to the fractal nature of the space. All that is to say, the unique properties of causal random geometry allow us to readily extend the arguments in \cite{Aasen_2016, aasen2020topological} to CTs but not to Euclidean triangulations. 

In the \textit{annealed model}, which amounts to sampling a graph $C \in \mathcal C_\infty$ according to the measure $\mu_J(C)$, we do not know the exact form of the process that describes the evolution of the random
geometry. However, we do know that it must be described by some one--dimensional stochastic process $X(t)$, due to the restriction to causal graphs with a global time foliation. It follows that, to extend our arguments to the annealed model, it would be sufficient to prove the continuity of $X(t)$. This argument applies to any coupling of unitary matter to CDT if a second order phase transition,  where a continuum limit can be defined, exists. 

Another interesting observation comes from the functional renormalization group equation (FRGE) analysis of the matrix model representation of CDT \cite{castro2020renormalization}. The authors find that the anomalous dimension vanishes due to the presence of the matrices $C$ (not to be confused with the causal graph) that impose the causal structure in the ribbon graphs. 
Since it is exactly the anomalous dimension that shifts the scaling exponents, it follows that there should be no KPZ--like relation. The same analysis of the Ising--CDT matrix model defined in \cite{Abranches} comes to the same conclusion \cite{Barouki:2025gky}. In fact, it is clear that any matrix model description of CDT coupled to matter that uses these $C$ matrices to impose the causal constraint will contain vanishing anomalous dimensions.

\acknowledgments
We thank Davide Laurenzano for insightful comments on the matrix model representation. We would also like to thank Isaac Layton for useful discussions and for introducing the Onsager-Machlup function. Ryan Barouki is supported by STFC studentship  ST/W507726/1    
and
Henry Stubbs is supported by STFC studentship  ST/X508664/1.  
For the purpose of open access, the authors have applied a CC BY public copyright
licence to any Author Accepted Manuscript (AAM) version arising from this submission.

\appendix
\section{Defects at the boundary}
\label{app:boundary_defects}
One is free to choose boundary conditions as part of the definition of the partition function. Two choices are \textit{free} and \textit{fixed}. In the free case, the boundary spins are summed over in the partition function just like any other bulk spin.
In the fixed case, the boundary is defined by a state, $\ket{B} = \ket{...0001101011...}$ for example, which specifies the spin at
every boundary vertex.  We are interested in whether a defect is free to move at the boundary without altering the partition
function.
Figure \ref{fig:spin_defect_boundary} shows the boundary spins $a,b,c,d$ in the presence of a spin defect. Notice that the final
plaquettes are horizontal to signify that an edge joins the boundary points.
\begin{figure}[h]
\centering
\begin{equation*}
\includegraphics[scale=1.5, valign=c]{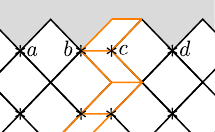}
\longrightarrow
\includegraphics[scale=1.5, valign=c]{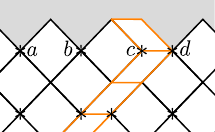}
\end{equation*}
\caption{Spin defect ending at the boundary (top) being moved one spin to the right. The boundary spins are those marked
$a,b,c,d$.} 
\label{fig:spin_defect_boundary}
\end{figure}

If the boundary conditions are free, then by the same bulk moves in (\ref{eq:defect_comm_1}--\ref{eq:defect_comm_2}), a spin defect move
at the boundary is also topological. On the other hand, suppose we fix the boundary such that $a = 0, b = 0, c = 1, d = 1$, then
before the move, the section of the boundary shown in figure \ref{fig:spin_defect_boundary} contributes a factor $W^H(u)_{00}
W^H(u)_{11} \times 1$, where the factor of 1 comes from the spin defect. After the move the contribution is $W^H(u)_{00}
W^H(u)_{10} \times 0$, where the factor of 0 also comes from the spin defect but with the same spin on each side of the plaquette.
Hence  this move does not preserve the partition function and is not topological. Moving the defect to the left is similarly not allowed. The defect can therefore be pinned by choosing a boundary condition in which all spins are $1$, except for two adjacent spins that are $0$.

\begin{figure}[h]
\centering
\begin{equation*}
\includegraphics[scale=0.25, valign=c]{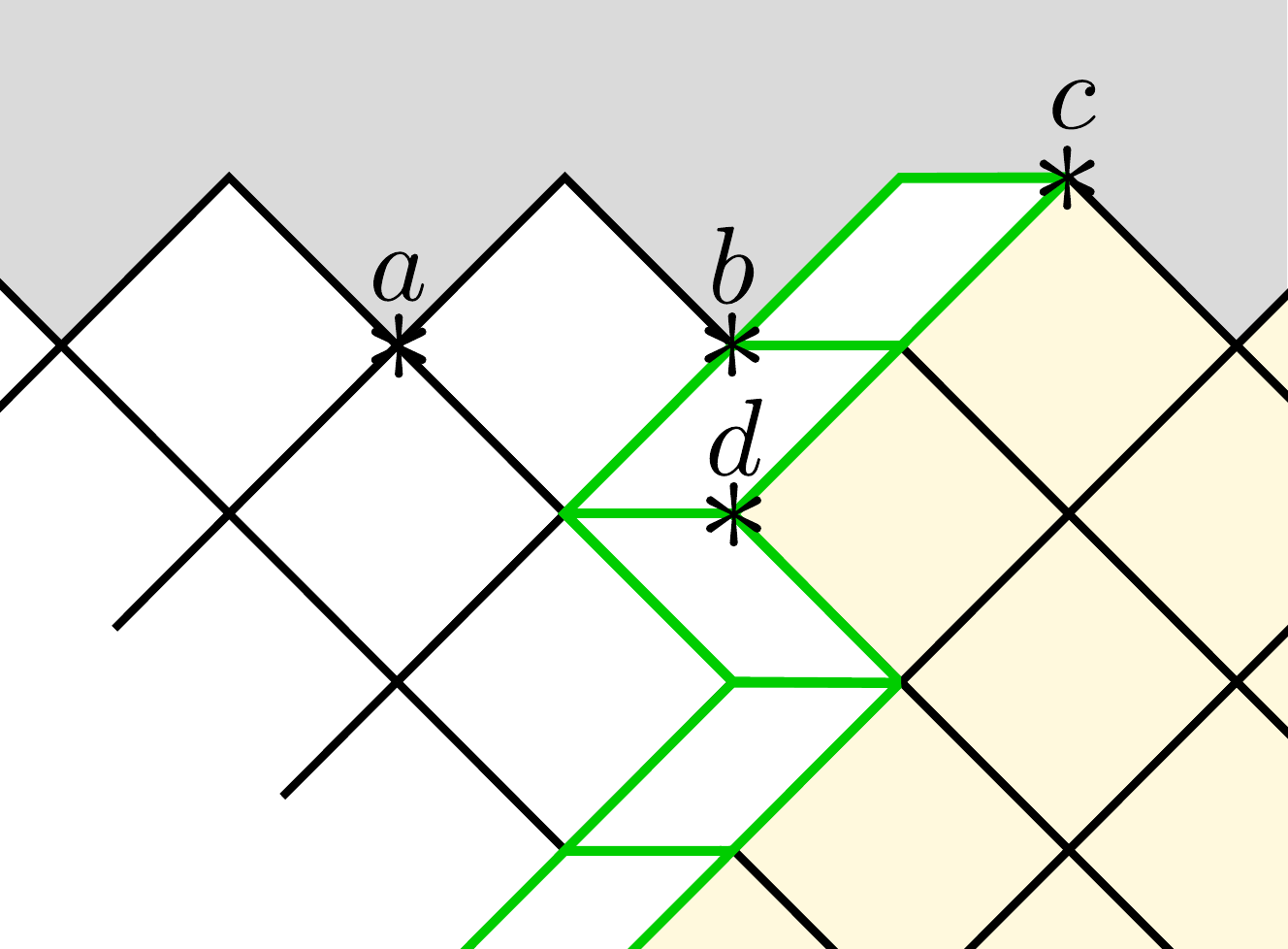}
\longrightarrow
\includegraphics[scale=0.25, valign=c]{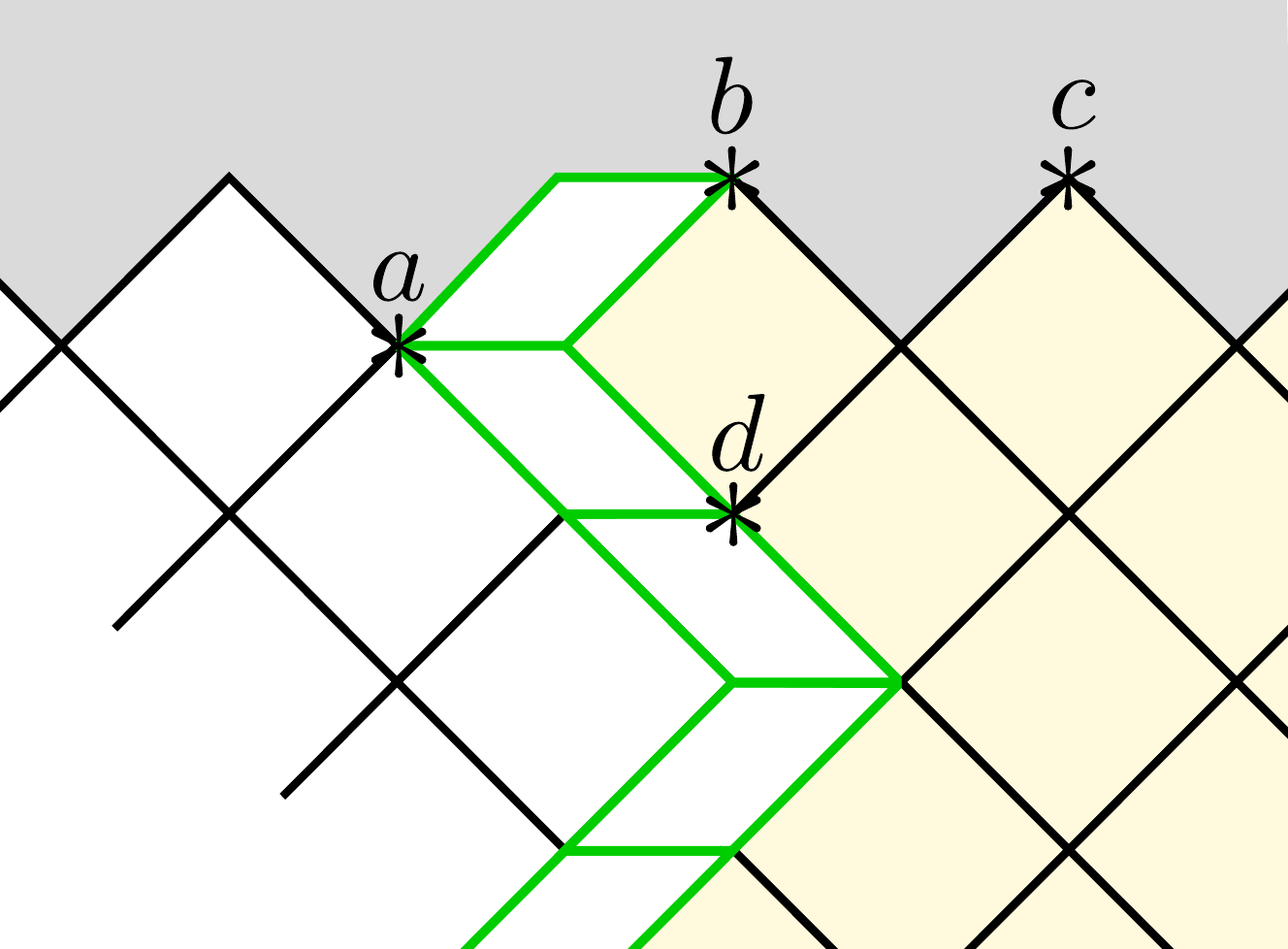}
\end{equation*}
\caption{Moving a duality defect at the boundary (top) one step to the left. The boundary spins are $a,b,c$. The yellow shading shows the difference between the dual and original plaquettes.}
\label{fig:duality_at_boundary}
\end{figure}

The case of the duality defect is slightly trickier -- consider the diagram in figure \ref{fig:duality_at_boundary}.
The left-hand side of the figure contributes a factor 
\begin{equation}
  W^H(u_H)_{ab}(-1)^{bd}(-1)^{bc},
  \label{eq:duality_boundary_lhs}
\end{equation}
whereas the right-hand side is given by
\begin{equation}
  W^V(u_H)_{bd}(-1)^{ad}(-1)^{ab}.
  \label{eq:duality_boundary_rhs}
\end{equation}
Now consider free boundary conditions. This amounts to summing over $b$ since this is the only internal spin in the
diagram. Evaluating $\sum_b$(\ref{eq:duality_boundary_lhs}) and $\sum_b$(\ref{eq:duality_boundary_rhs}), one finds that these can never be equal for all combinations of $a,d,c$.
Thus the duality defect cannot be moved along the boundary without changing the partition function for all values of $u_H, u_V$. 
A similar argument shows that for a magnetised boundary
with $a=b=c=0$ 
again 
the duality defect cannot be moved freely.

 We conclude  that it is possible to pin any defect to the boundary by choosing
suitable boundary conditions.

\section{The domain wall}
\label{app:wall}
As we alluded to in Section \ref{sec:top_ising}, when considering a vertical duality defect in a space with periodic boundary conditions, we
necessarily require a domain wall where the dual lattice meets the original again. In this appendix, we will elaborate on the
details of the wall and show that defects can pass through the wall topologically.

One way to implement the wall is to introduce a new plaquette that identifies the spins at the wall. 
This plaquette is defined as 

\begin{equation}
  \includegraphics[valign=c, scale=0.5]{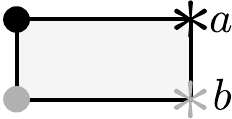} = \delta_{ab},
  \label{eq:wall_plaquette}
\end{equation}
and Figure \ref{fig:wall} shows how it is implemented at the interface of the original and dual lattices.

\begin{figure}
  \begin{center}
    \includegraphics[width=0.95\textwidth]{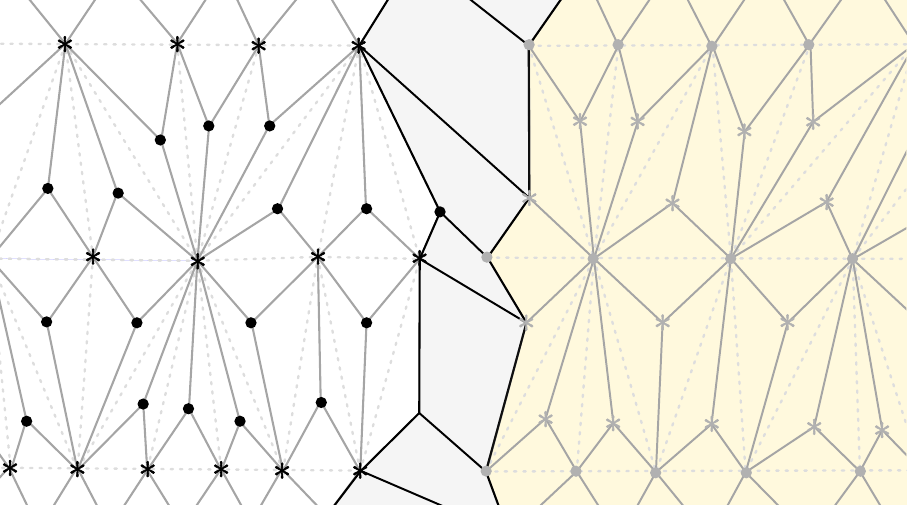}
  \end{center}
\caption{The domain wall where the original and dual lattice meet again in the presence of a vertical duality defect. Black asterisks indicate the positions of the spins on $G$, while gray asterisks mark the positions of the spins on $G^*$. Black dots denote empty positions in $G$, and gray dots denote empty positions in $G^*$. The dotted lines depict the underlying triangulation.}\label{fig:wall}
\end{figure}

Consider a horizontal sequence of duality defect plaquettes, such as those that constitute the Dehn twist in the presence of a
vertical duality defect. At some point, there will be an intersection of the horizontal plaquettes with the wall, as shown in
Figure \ref{fig:wall_duality}. Notice that at the intersection of the wall and the horizontal defect, there is a new type of object
\begin{equation}
  \includegraphics[valign=c, scale=0.5]{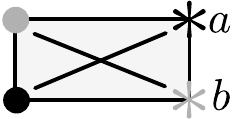} = \delta_{ab}
  \label{eq:wall_defect_plaquette}
\end{equation}
which is almost identical to (\ref{eq:wall_plaquette}) but differs with the relative positions of the original and dual empty sites. 
\begin{figure}
  \begin{center}
    \includegraphics[width=0.95\textwidth]{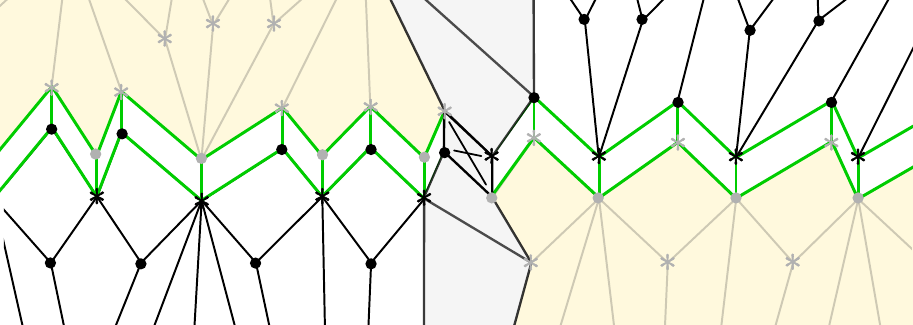}
  \end{center}
  \caption{The intersection of a horizontal duality defect and the domain wall.}\label{fig:wall_duality}
\end{figure}

The question is whether the horizontal defect can be moved along the wall and what happens when two
horizontal defects meet at the wall. To answer the first question, consider the following diagram
\begin{equation}
  \includegraphics[valign=c, width=0.3\textwidth]{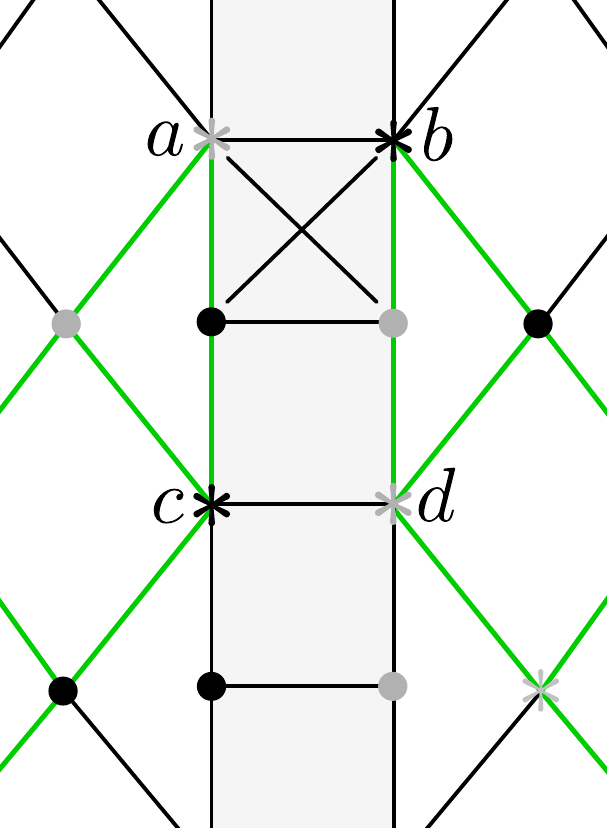}
  ~~~~
  \to
  ~~~~
  \includegraphics[valign=c, width=0.3\textwidth]{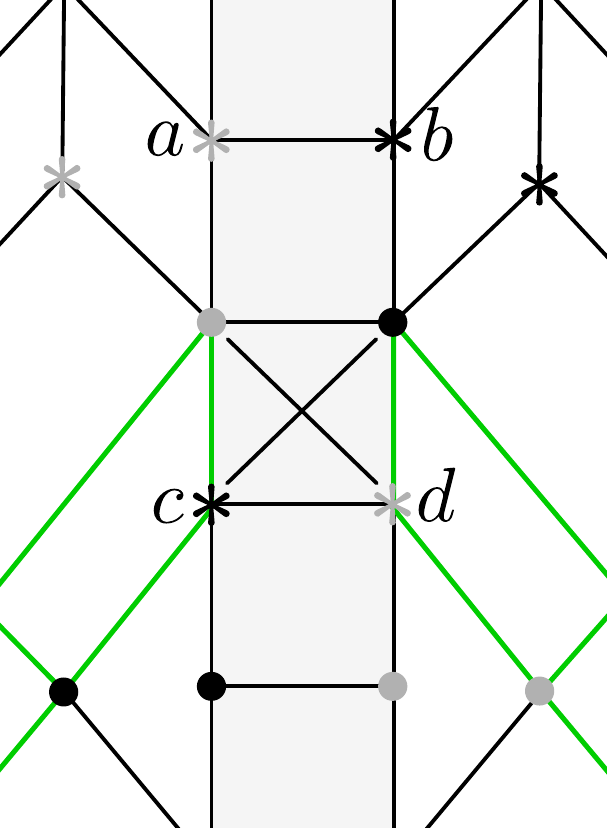}.
  \label{eq:duality_wall_move}
\end{equation}
Evaluating the left-hand side, 
\begin{equation}
\begin{aligned}
  &= 2^{-1/2}(-1)^{ac} \times 2^{-1/2}(-1)^{bd} \times \delta_{ab}\delta_{cd}\times 2^{1/2} \times 2^{1/2}\\
  &= (-1)^{ac}(-1)^{bd} \delta_{ab}\delta_{cd} = \delta_{ab}\delta_{cd}
\end{aligned}
\end{equation}
where the final factors of $\sqrt{2}$ come from the additional blank vertices on the left-hand side compared to the right-hand
side. The right-hand side of (\ref{eq:duality_wall_move}) gives the same result.
Therefore, we have shown that the duality defect can be moved along the wall without obstruction.
Finally, we must show what happens when two horizontal duality defects meet at the wall.
\begin{figure}
  \begin{center}
    \includegraphics[width=0.3\textwidth]{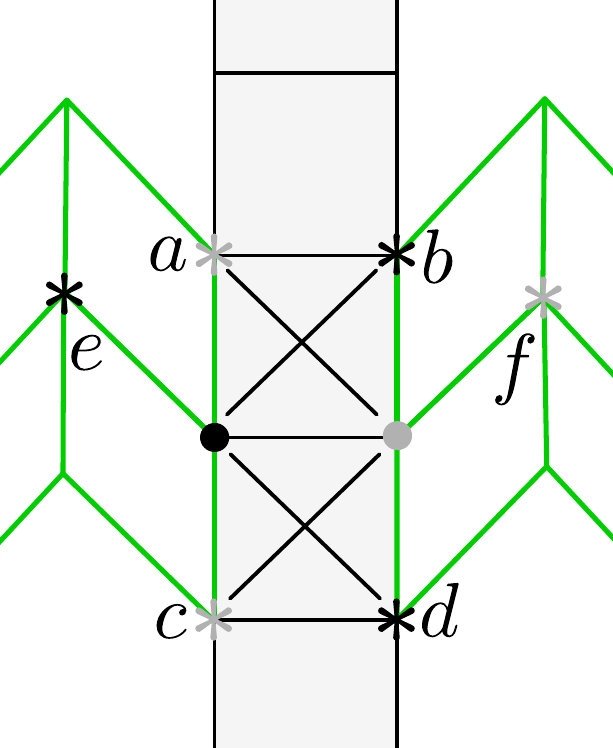}
  \end{center}
  \caption{Two horizontal duality defects meeting at the wall.}\label{fig:two_duality_wall}
\end{figure}
Consider the diagram in Figure \ref{fig:two_duality_wall}, as was shown in \cite{Aasen_2016} when two duality defects meet, it is
equivalent to the sum of an identity and a spin defect.

\begin{equation}
  \begin{split}
  \includegraphics[valign=c, width=0.17\textwidth]{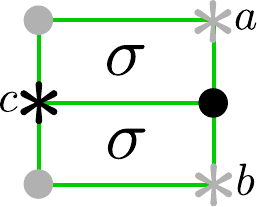}
  &= 
  \frac{1}{\sqrt{2}}\left[
  \includegraphics[valign=c, width=0.2\textwidth]{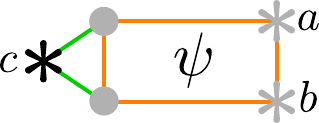}
  +
  \includegraphics[valign=c, width=0.2\textwidth]{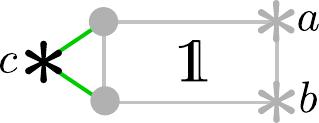}
  \right]\\
  &= 
  \frac{1}{\sqrt{2}}\left[\delta_{ab} + \sigma^x_{ab}\right]
\end{split}
  \label{eq:defect_to_sum_plaquette}
\end{equation}
Since the spins across the wall are identified, the only non-zero
contributions are those in which there is either an identity or a spin defect on both sides of Figure
\ref{fig:two_duality_wall}. Contributions where there is an identity on one side but a spin on the other vanish.

We are interested in whether we can perform the topological moves needed at the end of Section \ref{sec:top_ising} in the presence
of a wall. In this setup, a single duality defect crosses the wall multiple times, as shown in Figure \ref{fig:duality_crossing_wall}.
\begin{figure}
  \begin{center}
    \begin{equation*}
    \includegraphics[valign=c,width=0.4\textwidth]{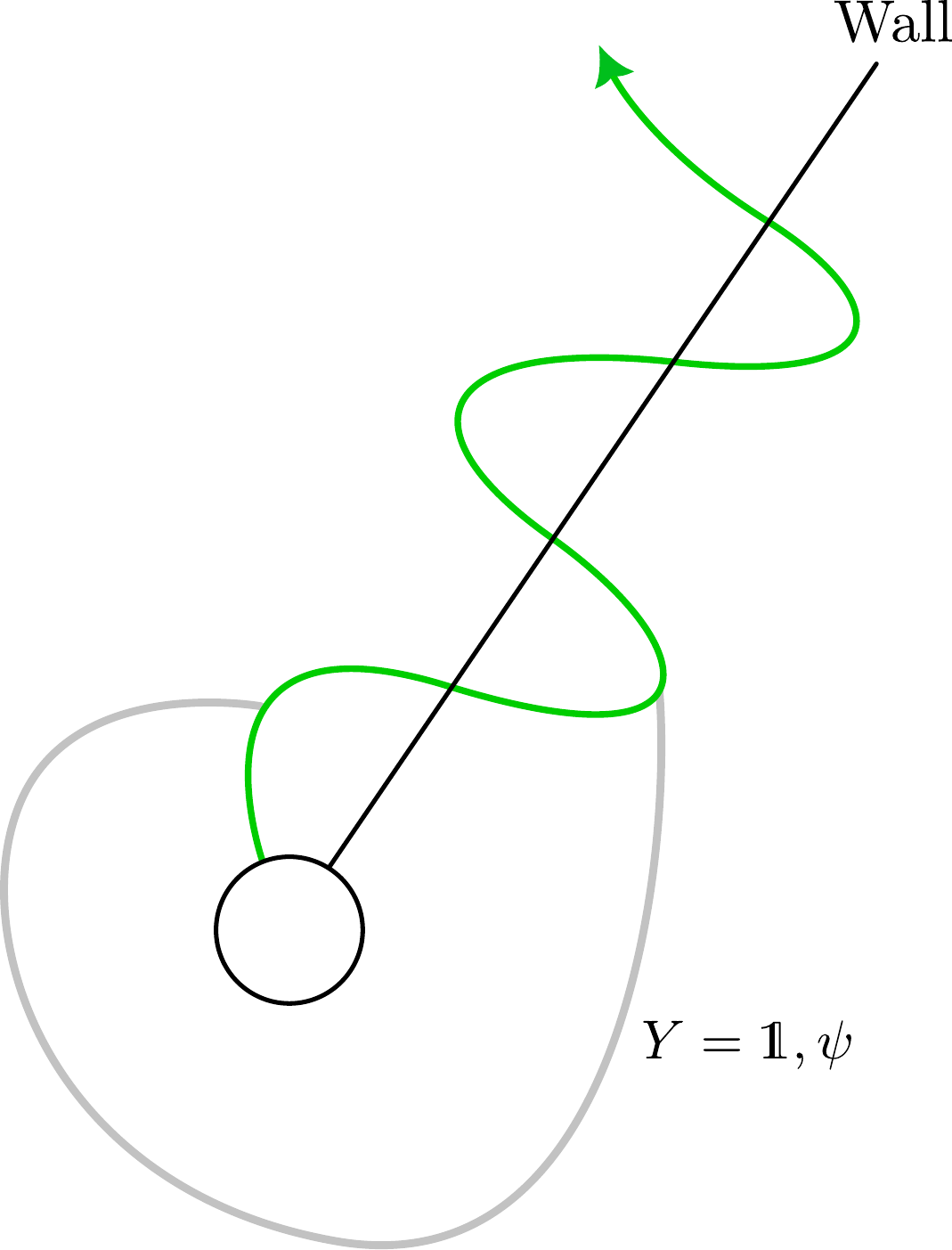}
    = 
    \sum_{X=\{\mathds{1}, \psi\}}
    \includegraphics[valign=c,width=0.4\textwidth]{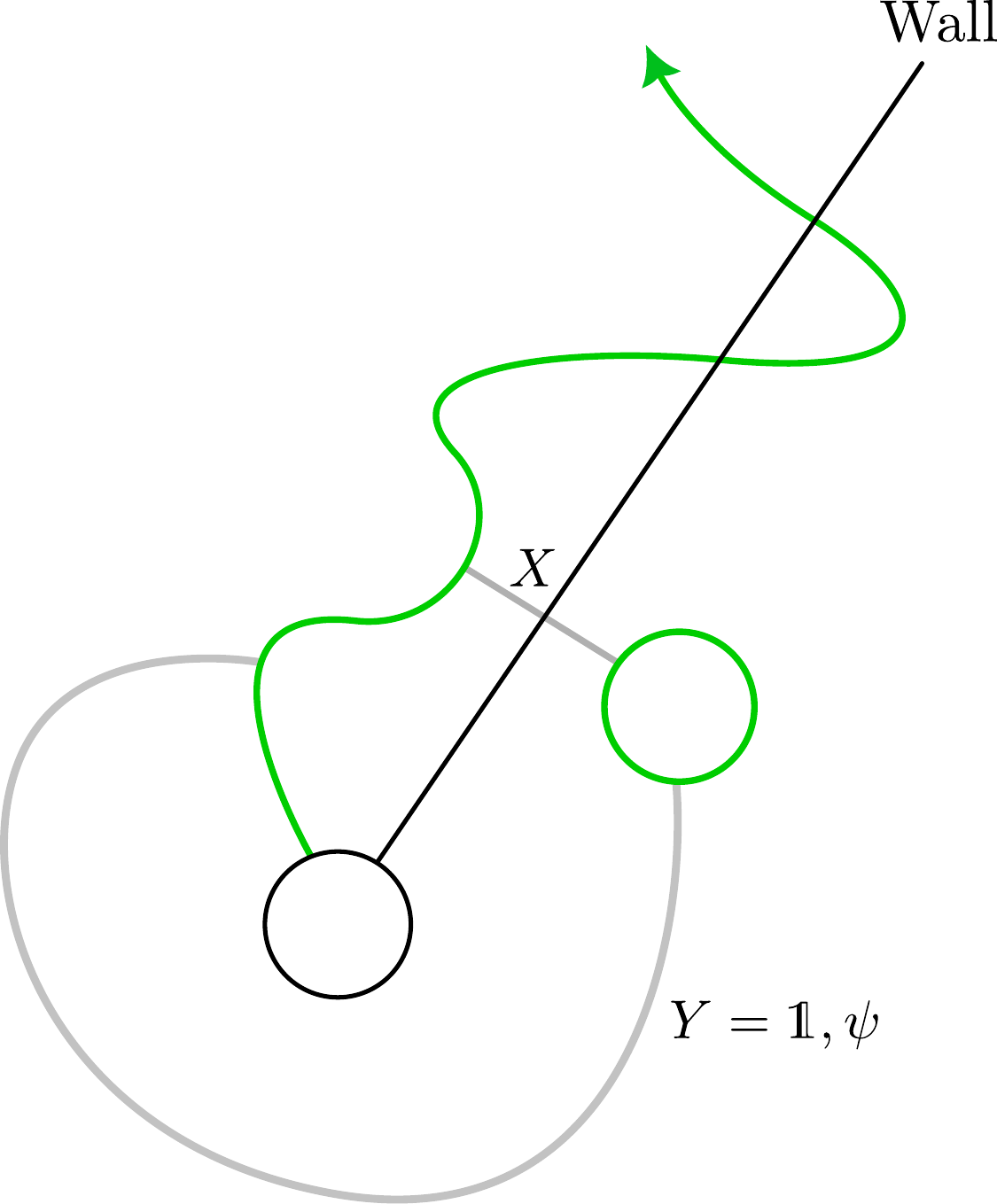}
  \end{equation*}
  \end{center}
  \caption{Macroscopic view of applying (\ref{eq:defect_to_sum_plaquette}) to the set up in Section \ref{sec:conformal_dims}. If $Y=\mathds{1}$,
  the non-zero contribution is the identity defect passing through the wall. Otherwise, if $Y=\psi$, then the spin defect passes through the wall.}\label{fig:duality_crossing_wall}
\end{figure}
As we concluded above, merging two duality lines at the wall produces a sum over an identity defect and a spin defect that
passes through the wall. Depending on the type of defect that wraps around the centre of the space, either the identity or the spin
defect will pass through the wall as shown in Figure \ref{fig:duality_crossing_wall}. In each case, the remaining duality bubble shrinks to zero, contributing a factor of $\sqrt{2}$ which cancels the factor in the sum in (\ref{eq:defect_to_sum_plaquette}).
Hence, we have shown that the duality defect can be moved through the wall, allowing us to perform the necessary moves in calculating
the conformal dimensions in Section \ref{sec:conformal_dims}.

\bibliographystyle{JHEP}
\bibliography{main.bib}

\providecommand{\href}[2]{#2}\begingroup\raggedright\begin{thebibliography}{10}

\bibitem{duplantier1998}
B.~Duplantier, \emph{Random walks and quantum gravity in two dimensions}, \href{https://doi.org/10.1103/PhysRevLett.81.5489}{\emph{Phys. Rev. Lett.} {\bfseries 81} (1998) 5489}.

\bibitem{Knizhnik:1988ak}
V.G.~Knizhnik, A.M.~Polyakov and A.B.~Zamolodchikov, \emph{{Fractal Structure of 2D Quantum Gravity}}, \href{https://doi.org/10.1142/S0217732388000982}{\emph{Mod. Phys. Lett. A} {\bfseries 3} (1988) 819}.

\bibitem{David:1988hj}
F.~David, \emph{{Conformal Field Theories Coupled to 2D Gravity in the Conformal Gauge}}, \href{https://doi.org/10.1142/S0217732388001975}{\emph{Mod. Phys. Lett. A} {\bfseries 3} (1988) 1651}.

\bibitem{DISTLER1989509}
J.~Distler and H.~Kawai, \emph{Conformal field theory and 2d quantum gravity}, \href{https://doi.org/https://doi.org/10.1016/0550-3213(89)90354-4}{\emph{Nuclear Physics B} {\bfseries 321} (1989) 509}.

\bibitem{kazakov_ising_1986}
V.~Kazakov, \emph{Ising model on a dynamical planar random lattice: {Exact} solution}, \href{https://doi.org/10.1016/0375-9601(86)90433-0}{\emph{Physics Letters A} {\bfseries 119} (1986) 140}.

\bibitem{Boulatov:1986sb}
D.V.~Boulatov and V.A.~Kazakov, \emph{{The Ising Model on Random Planar Lattice: The Structure of Phase Transition and the Exact Critical Exponents}}, \href{https://doi.org/10.1016/0370-2693(87)90312-1}{\emph{Phys. Lett. B} {\bfseries 186} (1987) 379}.

\bibitem{Ambj_rn_1998}
J.~Ambjørn and R.~Loll, \emph{Non-perturbative lorentzian quantum gravity, causality and topology change}, \href{https://doi.org/10.1016/s0550-3213(98)00692-0}{\emph{Nuclear Physics B} {\bfseries 536} (1998) 407–434}.

\bibitem{Ambj_rn_1999_numerics}
J.~Ambjørn, K.N.~Anagnostopoulos and R.~Loll, \emph{New perspective on matter coupling in 2d quantum gravity}, \href{https://doi.org/10.1103/physrevd.60.104035}{\emph{Physical Review D} {\bfseries 60} (1999) }.

\bibitem{Ambj_rn_2009}
J.~Ambj{\o}rn, K.~Anagnostopoulos, R.~Loll and I.~Pushkina, \emph{{Shaken, but not stirred{\textemdash}Potts model coupled to quantum gravity}}, \href{https://doi.org/10.1016/j.nuclphysb.2008.08.030}{\emph{Nuclear Physics B} {\bfseries 807} (2009) 251}.

\bibitem{Ambj_rn_2000}
J.~Ambj{\o}rn, K.~Anagnostopoulos and R.~Loll, \emph{Crossing the c=1 barrier in 2d lorentzian quantum gravity}, {\emph{Physical Review D} {\bfseries 61} (2000) 44010}.

\bibitem{Wheater_2022}
J.F.~Wheater and P.D.~Xavier, \emph{The cylinder amplitude in the hard dimer model on 2d causal dynamical triangulations}, \href{https://doi.org/10.1088/1361-6382/ac50ec}{\emph{Classical and Quantum Gravity} {\bfseries 39} (2022) 075004}.

\bibitem{Aasen_2016}
D.~Aasen, R.S.K.~Mong and P.~Fendley, \emph{{Topological defects on the lattice: I. The Ising model}}, \href{https://doi.org/10.1088/1751-8113/49/35/354001}{\emph{Journal of Physics A: Mathematical and Theoretical} {\bfseries 49} (2016) 354001}.

\bibitem{duplantier2010liouville}
B.~Duplantier and S.~Sheffield, \emph{Liouville quantum gravity and kpz}, \href{https://doi.org/10.1007/s00222-010-0308-1}{\emph{Inventiones mathematicae} {\bfseries 185} (2011) 333}.

\bibitem{Ambj_rn_2013}
J.~Ambjørn, L.~Glaser, Y.~Sato and Y.~Watabiki, \emph{2d cdt is 2d hořava–lifshitz quantum gravity}, \href{https://doi.org/10.1016/j.physletb.2013.04.006}{\emph{Physics Letters B} {\bfseries 722} (2013) 172–175}.

\bibitem{Regge1961GeneralRW}
T.E.~Regge, \emph{General relativity without coordinates}, {\emph{Il Nuovo Cimento (1955-1965)} {\bfseries 19} (1961) 558}.

\bibitem{Tutte1962a}
W.T.~Tutte, \emph{A census of hamiltonian polygons}, {\emph{Canadian Journal of Mathematics} {\bfseries 14} (1962) 402 }.

\bibitem{Tutte1962b}
W.T.~Tutte, \emph{A census of planar triangulations}, {\emph{Canadian Journal of Mathematics} {\bfseries 14} (1962) 21 }.

\bibitem{Tutte1962c}
W.T.~Tutte, \emph{A census of slicings}, {\emph{Canadian Journal of Mathematics} {\bfseries 14} (1962) 708 }.

\bibitem{Tutte1963}
W.T.~Tutte, \emph{A census of planar maps}, {\emph{Canadian Journal of Mathematics} {\bfseries 15} (1963) 249 }.

\bibitem{KAZAKOV1985295}
V.~Kazakov, I.~Kostov and A.~Migdal, \emph{Critical properties of randomly triangulated planar random surfaces}, \href{https://doi.org/https://doi.org/10.1016/0370-2693(85)90669-0}{\emph{Physics Letters B} {\bfseries 157} (1985) 295}.

\bibitem{DAVID1985543}
F.~David, \emph{A model of random surfaces with non-trivial critical behaviour}, \href{https://doi.org/https://doi.org/10.1016/0550-3213(85)90363-3}{\emph{Nuclear Physics B} {\bfseries 257} (1985) 543}.

\bibitem{Brezin:1977sv}
E.~Brezin, C.~Itzykson, G.~Parisi and J.B.~Zuber, \emph{{Planar Diagrams}}, \href{https://doi.org/10.1007/BF01614153}{\emph{Commun. Math. Phys.} {\bfseries 59} (1978) 35}.

\bibitem{Anninos:2020ccj}
D.~Anninos and B.~M\"uhlmann, \emph{{Notes on matrix models (matrix musings)}}, \href{https://doi.org/10.1088/1742-5468/aba499}{\emph{J. Stat. Mech.} {\bfseries 2008} (2020) 083109} [\href{https://arxiv.org/abs/2004.01171}{{\ttfamily 2004.01171}}].

\bibitem{gall_hausdorff}
J.-F.~Le~Gall, \emph{The topological structure of scaling limits of large planar maps}, \href{https://doi.org/10.1007/s00222-007-0059-9}{\emph{Inventiones mathematicae} {\bfseries 169} (2007) 621}.

\bibitem{gall_brownian_map}
J.-F.~Le~Gall, \emph{Uniqueness and universality of the {Brownian} map}, \href{https://doi.org/10.1214/12-AOP792}{\emph{The Annals of Probability} {\bfseries 41} (2013) 2880}.

\bibitem{JerBet}
J.~Bettinelli, ``Planar map with 30,000 vertices.''

\bibitem{Sorkin_1975}
R.~Sorkin, \emph{Time-evolution problem in regge calculus}, \href{https://doi.org/10.1103/PhysRevD.12.385}{\emph{Phys. Rev. D} {\bfseries 12} (1975) 385}.

\bibitem{Ambjorn:2013joa}
J.~Ambj\o{}rn, L.~Glaser, Y.~Sato and Y.~Watabiki, \emph{{2d CDT is 2d Ho\v{r}ava\textendash{}Lifshitz quantum gravity}}, \href{https://doi.org/10.1016/j.physletb.2013.04.006}{\emph{Phys. Lett. B} {\bfseries 722} (2013) 172} [\href{https://arxiv.org/abs/1302.6359}{{\ttfamily 1302.6359}}].

\bibitem{tree_bijection}
M.~Krikun and A.~Yambartsev, \emph{Phase transition for the ising model on the critical lorentzian triangulation}, \href{https://doi.org/https://doi.org/10.1007/s10955-012-0548-0}{\emph{J Stat Phys} {\bfseries 148} (2008) 422–439} [\href{https://arxiv.org/abs/0810.2182}{{\ttfamily 0810.2182}}].

\bibitem{durhuus2022trees}
B.~Durhuus, T.~Jonsson and J.~Wheater, \emph{From trees to gravity},  in \emph{Handbook of Quantum Gravity}, C.~Bambi, L.~Modesto and I.~Shapiro, eds., (Singapore), pp.~3385--3435, Springer Nature Singapore (2024), \href{https://doi.org/10.1007/978-981-99-7681-2_86}{DOI}.

\bibitem{Durhuus:2009sm}
B.~Durhuus, T.~Jonsson and J.F.~Wheater, \emph{{On the spectral dimension of causal triangulations}}, \href{https://doi.org/10.1007/s10955-010-9968-x}{\emph{J. Statist. Phys.} {\bfseries 139} (2010) 859} [\href{https://arxiv.org/abs/0908.3643}{{\ttfamily 0908.3643}}].

\bibitem{durhuus_spine2003}
B.~{Durhuus}, \emph{{Probabilistic Aspects of Infinite Trees and Surfaces}}, {\emph{Acta Physica Polonica B} {\bfseries 34} (2003) 4795}.

\bibitem{POLYAKOV1981207}
A.~Polyakov, \emph{Quantum geometry of bosonic strings}, \href{https://doi.org/https://doi.org/10.1016/0370-2693(81)90743-7}{\emph{Physics Letters B} {\bfseries 103} (1981) 207}.

\bibitem{Teschner_2001}
J.~Teschner, \emph{Liouville theory revisited}, \href{https://doi.org/10.1088/0264-9381/18/23/201}{\emph{Classical and Quantum Gravity} {\bfseries 18} (2001) R153–R222}.

\bibitem{Nakayama:2004vk}
Y.~Nakayama, \emph{{Liouville field theory: A Decade after the revolution}}, \href{https://doi.org/10.1142/S0217751X04019500}{\emph{Int. J. Mod. Phys. A} {\bfseries 19} (2004) 2771} [\href{https://arxiv.org/abs/hep-th/0402009}{{\ttfamily hep-th/0402009}}].

\bibitem{Simon_1975}
F.~Guerra, L.~Rosen and B.~Simon, \emph{The $p(\phi)_2$ euclidean quantum field theory as classical statistical mechanics}, {\emph{Annals of Mathematics} {\bfseries 101} (1975) 111}.

\bibitem{aasen2020topological}
D.~Aasen, P.~Fendley and R.S.K.~Mong, \emph{Topological defects on the lattice: Dualities and degeneracies},  \href{https://arxiv.org/abs/2008.08598}{{\ttfamily 2008.08598}}.

\bibitem{hernandez_bounds_2013}
J.C.~Hernandez, Y.~Suhov, A.~Yambartsev and S.~Zohren, \emph{Bounds on the critical line via transfer matrix methods for an {Ising} model coupled to causal dynamical triangulations}, \href{https://doi.org/10.1063/1.4808101}{\emph{Journal of Mathematical Physics} {\bfseries 54} (2013) 063301}.

\bibitem{lamperti_ney1968}
J.~Lamperti and P.~Ney, \emph{{Conditioned Branching Processes and Their Limiting Diffusions}}, \href{https://doi.org/10.1137/1113009}{\emph{Theory of Probability \& Its Applications} {\bfseries 13} (1968) 128} [\href{https://arxiv.org/abs/https://doi.org/10.1137/1113009}{{\ttfamily https://doi.org/10.1137/1113009}}].

\bibitem{sisko2011note}
V.~Sisko, A.~Yambartsev and S.~Zohren, \emph{{A note on weak convergence results for infinite causal triangulations}}, \href{https://doi.org/10.1214/17-BJPS356}{\emph{Brazilian Journal of Probability and Statistics} {\bfseries 32} (2018) 597 }.

\bibitem{ethier_kurtz_book}
S.N.~Ethier and T.G.~Kurtz, \emph{Markov processes -- characterization and convergence}, Wiley Series in Probability and Mathematical Statistics: Probability and Mathematical Statistics, John Wiley \& Sons Inc., New York (1986).

\bibitem{kurtz_protter1991}
T.G.~Kurtz and P.~Protter, \emph{{Weak Limit Theorems for Stochastic Integrals and Stochastic Differential Equations}}, \href{https://doi.org/10.1214/aop/1176990334}{\emph{The Annals of Probability} {\bfseries 19} (1991) 1035 }.

\bibitem{RevuzMartingales}
D.~Revuz and M.~Yor, \emph{Continuous Martingales and Brownian Motion}, Grundlehren der mathematischen Wissenschaften, A Series of Comprehensive Studies in Mathematics, 293, Springer Berlin Heidelberg, Berlin, Heidelberg, 3rd ed. 1999.~ed. (1999).

\bibitem{onsager_machlup}
L.~Onsager and S.~Machlup, \emph{Fluctuations and irreversible processes}, \href{https://doi.org/10.1103/PhysRev.91.1505}{\emph{Phys. Rev.} {\bfseries 91} (1953) 1505}.

\bibitem{bach_hors1975}
W.~{Horsthemke} and A.~{Bach}, \emph{{Onsager-Machlup Function for one dimensional nonlinear diffusion processes}}, \href{https://doi.org/10.1007/BF01322364}{\emph{Zeitschrift fur Physik B Condensed Matter} {\bfseries 22} (1975) 189}.

\bibitem{Weber_2017}
M.F.~Weber and E.~Frey, \emph{Master equations and the theory of stochastic path integrals}, \href{https://doi.org/10.1088/1361-6633/aa5ae2}{\emph{Reports on Progress in Physics} {\bfseries 80} (2017) 046601}.

\bibitem{Ambj_rn_1999}
J.~Ambj{\o}rn, K.N.~Anagnostopoulos and R.~Loll, \emph{New perspective on matter coupling in 2d quantum gravity}, \href{https://doi.org/10.1103/physrevd.60.104035}{\emph{Physical Review D} {\bfseries 60} (1999) }.

\bibitem{Benedetti_2007}
D.~Benedetti and R.~Loll, \emph{Unexpected spin-off from quantum gravity}, \href{https://doi.org/10.1016/j.physa.2006.11.032}{\emph{Physica A: Statistical Mechanics and its Applications} {\bfseries 377} (2007) 373}.

\bibitem{castro2020renormalization}
A.~Castro and T.~Koslowski, \emph{Renormalization group approach to the continuum limit of matrix models of quantum gravity with preferred foliation}, \href{https://doi.org/10.3389/fphy.2021.531766}{\emph{Frontiers in Physics} {\bfseries 9} (2021) } [\href{https://arxiv.org/abs/2008.10090}{{\ttfamily 2008.10090}}].

\bibitem{Abranches}
J.L.A.~Abranches, A.D.~Pereira and R.~Toriumi, \emph{Dually weighted multi-matrix models as a path to causal gravity-matter systems}, \href{https://doi.org/10.1007/s00023-024-01442-1}{\emph{Annales Henri Poincar{\'e}} (2024) } [\href{https://arxiv.org/abs/2310.13503}{{\ttfamily 2310.13503}}].

\bibitem{Barouki:2025gky}
R.~Barouki and D.~Laurenzano, \emph{{A Renormalization Group Analysis of the Ising Model Coupled to Causal Dynamical Triangulations}},  \href{https://arxiv.org/abs/2504.01134}{{\ttfamily 2504.01134}}.

\end{thebibliography}\endgroup
\end{document}